\documentclass[a4paper,UKenglish,cleveref, autoref, thm-restate]{lmcs}
\pdfoutput=1
\usepackage[utf8]{inputenc}

\usepackage{lastpage}
\lmcsdoi{21}{4}{5}
\lmcsheading{}{\pageref{LastPage}}{}{}%
{Nov.~04,~2024}{Oct.~08,~2025}{}

\usepackage{amsmath}
\usepackage{amsfonts}
\usepackage{amsthm}
\usepackage{tikz-cd}
\usepackage{stmaryrd}
\usepackage{bussproofs}
\usepackage{color}
\usepackage{mathpartir}
\usepackage{hyperref}
\usepackage{wrapfig}

\newtheorem{convention}[thm]{Convention}

\newcommand{\bb}[1]{\mathbb{#1}}

\newcommand{\catset}{\mathbf{Set}}

\newcommand{\grD}{{D^{\kappa}}}
\newcommand{\grDsc}{D^{\mathsf{g}}}
\newcommand{\ciD}{D}
\newcommand{\ciDsetoid}{D_{\sym{sd}}}
\newcommand{\Tsetoid}{T_{\sym{sd}}}
\newcommand{\EM}[1]{\mathcal{C}^{#1}}
\newcommand{\interpretation}[2]{{#1}_{#2}}
\newcommand{\seqinterpretation}[2]{{#1}^{seq}_{#2}}
\newcommand{\parinterpretation}[2]{{#1}^{par}_{#2}}

\renewcommand{\tt}{\mathsf{tt}}
\newcommand{\ff}{\mathsf{ff}}

\newcommand{\distlaw}{\zeta}

\DeclareMathOperator{\nextop}{next}
\DeclareMathOperator{\now}{now}
\DeclareMathOperator{\step}{step}
\DeclareMathOperator{\var}{var}
\DeclareMathOperator{\ar}{ar}

\DeclareMathOperator{\inl}{inl}
\DeclareMathOperator{\inr}{inr}
\DeclareMathOperator{\id}{Id}

\DeclareMathOperator{\binop}{op}
\DeclareMathOperator{\binopOne}{op_1}

\DeclareMathOperator{\binopPrime}{op'}
\DeclareMathOperator{\binopPrimePrime}{op''}

\DeclareMathOperator{\nop}{op} % n-ary operator

\DeclareMathOperator{\lookup}{lookup}
\DeclareMathOperator{\update}{update}
\DeclareMathOperator{\lookupspec}{lookup_{DS}}
\DeclareMathOperator{\updatespec}{update_{DS}}
\DeclareMathOperator{\stepspec}{step_{DS}}

\newcommand{\DistChoiceOp}{\oplus}% ^{\Dist}}
\NewDocumentCommand{\DistChoice}{mmm}{{#2}\DistChoiceOp_{#1}{#3}}
\newcommand{\II}{(0,1)}
\newcommand{\IIc}{[0,1]}

\newcommand{\iso}{\cong}
\newcommand{\defeq}{\mathbin{\overset{\textsf{def}}{=}}}

\newcommand{\later}{\triangleright}

\newcommand\tabs[2]{\lambda (#1\! :\! #2).}
\newcommand\tabsshort[2]{\lambda #1.}

\newcommand\tapp[2][\tickA]{#2\,[#1] }

\newcommand{\tickA}{\alpha}
\newcommand{\tickB}{\beta}
\newcommand\latbind[2]{{\triangleright}\, (#1 \!: \!#2) .}
\newcommand\latbindsmall[2]{{\triangleright}\, (#1 : #2) .}
\newcommand\toksubst[3][\kappa]{\left[#2/#3\right]}

\newcommand{\capp}[2][\kappa]{#2\,[#1]}

\newcommand{\clocktype}{\mathsf{clock}}

\newcommand{\fix}[1][\kappa]{\mathsf{fix}^{#1}}

\newcommand{\Set}{\mathsf{Set}}
\newcommand{\opcat}[1]{{{#1}^{\mathrm{op}}}}

\newcommand{\wbsimgr}[1][]{\sim_{#1}^{\kappa}\,}
\newcommand{\wbsim}[1][]{\sim_{#1}\,}
\newcommand{\Prop}{\mathsf{Prop}}
\newcommand{\N}{\mathbb{N}}
\newcommand{\eq}{\mathsf{eq}}

\newcommand{\divergence}{\mathsf{div}}

\newcommand\hastype[4][]{#2 \vdash_{#1} #3: #4}

\newcommand\wfcxt[2][]{#2 \vdash_{#1}}
\newcommand\istype[3][]{\ensuremath{#2 \vdash_{#1} #3 \, \operatorname{type}}}

\newcommand{\timeless}[1]{{\mathsf{TimeLess}(#1)}}

\newcommand{\subst}[2]{[#1/#2]}

\newcommand{\peq}{=}
\newcommand{\jeq}{\equiv}

\newcommand{\equi}{\simeq}

\newcommand\sym[1]{\mathsf{#1}}

\newcommand{\tirrAx}[1][\kappa]{\sym{tirr}^{#1}}
\newcommand{\pfix}[1][\kappa]{\sym{pfix}^{#1}}
\newcommand{\dfix}[1][\kappa]{\sym{dfix}^{#1}}

\newcommand{\Path}[3]{#2 \peq_{#1} #3}

\newcommand{\USet}{\sym{Set}}
\newcommand{\UhSet}{\sym{hSet}}
\newcommand{\UhProp}{\sym{hProp}}

\newcommand{\monadsum}[2]{#1\oplus#2}

\newcommand{\nextstep}[1][\kappa]{\delta^{#1}}

\newcommand{\Pfin}{{\mathsf{P}_{\mathsf{f}}}}
\newcommand{\Dfin}{{\mathsf{D}_{\mathsf{f}}}}

\newcommand{\clockc}{\kappa_0}
\newcommand{\cappc}[1]{\capp[\clockc] #1}

\begin{document}

% affiliations are numbered automatically with a, b, c (see below)
% use the optional argument to indicate the affiliation(s) of each author
% omit the argument if there is only one author, or only one affiliation
\author[R.~E.~M{\o}gelberg]{Rasmus Ejlers M{\o}gelberg\lmcsorcid{0000-0003-0386-4376}}
\author[M.~Zwart]{Maaike Zwart\lmcsorcid{0000-0002-0257-1574}}

% affiliation 1 (automatically numbered a)
\address{IT University of Copenhagen, Denmark}	%optional
% write emails for all authors having that affiliation
\email{mogel@itu.uk, maaike.annebeth@gmail.com}  %optional

\title[What Monads Can and Cannot Do]{What Monads Can and Cannot Do \texorpdfstring{\\}{} With a Few Extra Pages}
\titlecomment{{\lsuper*}This is an extended version of ``What Monads Can and Cannot Do With a Bit of Extra Time'' \cite{MogelbergZwart2024}}
\thanks{This work was funded by the Independent Research Fund Denmark, grant number 2032-00134B}	%optional

 %mandatory; please add comma-separated list of keywords

\begin{abstract}
The delay monad provides a way to introduce general recursion in type theory.
To write programs that use a wide range of computational effects directly in type theory, we need to combine the delay monad with the monads of these effects. Here we present a first systematic study of such combinations.

We study both the coinductive delay monad and its guarded recursive cousin, giving concrete examples of combining these with well-known computational effects. We also provide general theorems stating which algebraic effects distribute over the delay monad, and which do not. Lastly, we salvage some of the impossible cases by considering distributive laws up to weak bisimilarity.
\end{abstract}

\maketitle

\section{Introduction}

Martin L{\"o}f type theory~\cite{MartinLof:84} is a language that can be understood both
as a logic and a programming language. For the logical interpretation it is crucial that all
programs terminate. Still, one would like to reason about programming languages with
general recursion, or even write general recursive programs inside type theory.
One solution to this problem is to encapsulate recursion in a monad, such as
the delay monad $\ciD$. This monad maps an object $X$ to the coinductive solution to
$\ciD X \iso X + \ciD X$. The right inclusion into the sum of the above isomorphism introduces a computation
step, and infinitely many steps correspond to divergence. Capretta~\cite{capretta2005general} showed how $\ciD$
introduces general recursion via an iteration operator of type $(X \to \ciD (X + Y)) \to X \to \ciD Y$. For this
reason, $\ciD$ has been used to model recursion in type theory~\cite{NAD:SIGPLAN:Not:2012,VVMFPS2021,mococa},
and in particular forms part of the basis of interaction trees~\cite{Interaction:trees}.

The delay monad has a guarded recursive variant defined using Nakano's~\cite{Nakano:Modality}
fixed point modality $\later$. Data of type $\later X$ should be thought of as data of type $X$ available only one time step
from now. This modal operator has a unit $\nextop : X \to \later X$ transporting data to the future,
and a fixed point operator $\fix[] : (\later X \to X) \to X$ mapping productive definitions to their
fixed points satisfying $\fix[](f) = f(\nextop(\fix[] (f)))$. Guarded recursion can be modelled in the
topos of trees~\cite{ToT} -- the category $\Set^{\opcat{\omega}}$ of presheaves on the ordered
natural numbers~-- by defining
$(\later X)(0) = 1$ and $(\later X)(n+1) = X(n)$.
The guarded delay monad $\grDsc$ is defined as the free monad on $\later$, i.e., the inductive
(and provably also coinductive) solution to $\grDsc X \iso X + \later(\grDsc X)$.

The two delay monads can be formally related by moving to a multiclock variant of guarded recursion,
in which the modal operator $\later$ is indexed by a clock variable $\kappa$, which can be universally
quantified. Then by defining the guarded delay monad $\grD X$ to be the unique solution to $\grD X \iso
X + \later^\kappa(\grD X)$, the coinductive variant can be encoded~\cite{atkey13icfp} as $\ciD X \defeq \forall\kappa . \grD X$.
In this paper we will work informally in Clocked Cubical Type Theory (CCTT)~\cite{CubicalCloTT},
a type theory in which such encodings of coinductive types can be formalised and proven correct.

Unlike the coinductive variant, the guarded delay monad has a fixed point operator of type
$((X \to \grD Y) \to (X \to \grD Y)) \to (X \to \grD Y)$, defined using $\fix[]$. For the coinductive delay monad, fixed points only exist for continuous
maps, but in the guarded case, continuity is a consequence of a causality property enforced in
types using $\later$. As a consequence, higher order functional programming languages with
recursion can be embedded in type theories like CCTT
by interpreting function spaces as Kleisli exponentials for
$\grD$.
For example, Paviotti et al.~\cite{paviottiPCF}
showed how to model the simply typed lambda calculus with fixed point terms (PCF),
and proved adequacy of the model, all in a type theory with guarded recursion.
These results have since been extended to languages
with recursive types~\cite{Paviotti:FPC:journal} and
(using an impredicative universe) languages with higher-order
store~\cite{sterling-gratzer-birkedal:2022}. This suggests that the guarded
delay monad can be used for programming and reasoning about programs using a wide
range of advanced computational effects directly in type theory. However, a mathematical
theory describing the interaction of the delay monads with other monads is still lacking, even for
basic computational effects.

\subsection*{Combining the Delay Monad With Other Effects}

In this paper, we present a first systematic study of monads combining delay with other effects.
We first show how to combine the delay monad with standard monads known from computational
effects: exceptions, reader, global state, continuation, and the selection monad. Most of these follow standard combinations of effects and non-termination
known from domain theory, but, the algebraic status of these combinations is simpler than in the domain theoretic
case: Whereas the latter can be understood as free monads for order-enriched algebraic theories~\cite{hyland2006discrete}, the combinations with
the guarded recursive delay monad are simply free models of theories in the standard sense, with the caveat that the arity of the step
operation is non-standard.

The rest of the paper is concerned with distributive laws of the form $T \ciD \to \ciD T$, where
$T$ is any monad and $\ciD$ is the delay monad in any of the two forms mentioned above.
Such a distributive law distributes the operations of $T$ over steps, and equips the composite
$\ciD T$ with a monad structure that describes computations whose other effects are only visible
upon termination. This is the natural monad for example in the case of writing to state,
when considering non-determinism and observing must-termination, or for computing data contained in data structures such as trees or lists.

There are two natural ways of distributing an $n$-ary operation $\nop$ over computation steps:
The first is to execute each of the $n$ input computations in sequence
until they have all terminated, the second is to execute the $n$ inputs in parallel, delaying terminated
computations until all inputs have terminated. We show that sequential execution yields a distributive
law for algebraic monads (monads generated by algebraic theories) where all equations
are balanced, i.e., the number of occurrences of each variable on either side is the same.
Trees, lists, and multisets are examples of such monads.

The requirement of balanced equations is indeed necessary. This was observed already by
M{\o}gelberg and Vezzosi~\cite{mogelberg2021two}
who showed that for the finite powerset monad, sequential distribution of steps over the
union operator was not well defined, and parallel
distribution did not yield a distributive law due to miscounting
of steps. Here we strengthen this result to show that no distributive law
is possible for the finite powerset monad over the coinductive delay monad.
At first sight it might seem that the culprit in
this case is the idempotency axiom. However, we show that in some cases
it is possible to distribute
idempotent operations over the coinductive delay monad,
but not over the guarded one (we show this just for commutative operations).

Finally, we show that if one is willing to work up to weak bisimilarity, i.e., equating elements
of the delay monad that only differ by a finite number of computation steps, then
one can construct a distributive law $T\ciD \to \ciD T$ for any monad $T$
generated by an algebraic theory with no drop equations (equations where
a variable appears on one side, but not the other). To make this precise,
we formulate this result as a distributive law of monads on a category of setoids.

\paragraph*{Additional Material}
This paper is an extended version of the CSL 2024 conference paper \emph{What Monads Can and Cannot Do with a Bit of Extra Time} \cite{MogelbergZwart2024}. 
Compared to the conference paper, this version additionally includes:
\begin{itemize}
  \item A discussion of possible combinations of the delay monad with the continuations monad in Section~\ref{sec:examples}.
  \item A more extensive discussion of the distributive law for the delay monad over the selection monad in Section~\ref{sec:examples}, including a comparison to $T$-selection functions.
  \item A short paragraph about a distributive law for the writer monad over the delay monad in Section~\ref{sec:examples}.
  \item A new proposition, Proposition~\ref{prop:nodistprobability}, showing that there is no distributive law of the finite distributions functor over the coinductive delay monad, in Section~\ref{sec:nogo-idemp}.
  \item Detailed proofs of all lemmas, propositions and theorems in Sections~\ref{sec:examples}, \ref{sec:parseq}, \ref{sec:nogo-idemp}, and \ref{sec:weak:bisim}, including the proof of our main no-go theorem (Theorem~\ref{thm:gen:no:go}).
\end{itemize}

\paragraph*{Agda Formalisation}

Some of the results presented in this paper have been formalised in Agda \cite{AgdaFiles},
using Vezzosi's Guarded Cubical library \cite{AgdaCubical}.

\section{Monads and Algebraic Theories}\label{sec:prelim:monads}

In this background section we briefly remind the reader of algebraic theories and free model monads for an algebraic theory. We mention different classes of equations that play a role in our analysis of monad compositions, and we discuss distributive laws for composing monads.

\begin{defi}[Algebraic Theory]
An algebraic theory $A$ consists of a signature $\Sigma_A$ and a set of equations $E_A$. The signature is a set of operation symbols with arities given by natural numbers. The signature $\Sigma_A$ together with a set of variables $X$ inductively defines the set of $A$-terms: Every variable $x : X$ is a term, and for
each operation symbol $\nop$ in $\Sigma_A$, if $\nop$ has arity $n$ and $t_1, \ldots, t_n$ are terms, then $\nop(t_1, \ldots, t_n)$ is a term.
The set of equations contains pairs of terms $(s,t)$ in a finite variable context which are to be considered equal. We write $=$ for the smallest 
congruence relation generated by the pairs in $E_A$.
\end{defi}

\begin{exa}[Monoids]
  The algebraic theory of monoids has a signature consisting of a constant $c$ and a binary operation $*$, satisfying the left and right unital equations: $\forall x : X .\, c * x = x$ and $\forall x : X .\, x * c = x$, and associativity: $\forall x, y, z : X .\, (x * y) * z = x * (y * z) $. Commutative monoids also include the commutativity equation: $\forall x, y : X .\, x * y = y * x$. Idempotent commutative monoids, also known as semilattices, are commutative monads that additionally satisfy the idempotence equation: $\forall x : X.\, x * x = x$.
\end{exa}

\begin{exa}[Convex algebras]
 Convex algebras have a signature consisting of a binary operation $\DistChoice{p}{}{}$ for each $p \in \II$ satisfying:
 \begin{align*}
	  \tag{idempotence}  \DistChoice{p}{x}{x} &\peq x \\
	  \tag{commutativity} \DistChoice{p}{x}{y} &\peq \DistChoice{1-p}{y}{x} \\
	  \tag{associativity}  \DistChoice{q}{\left(\DistChoice{p}{x}{y}\right)}{z}
	  &\peq \DistChoice{pq}{x}{\left(\DistChoice{\frac{q-pq}{1-pq}}{y}{z}\right)}
  \end{align*}
\end{exa}

\begin{defi}[Category of Models]
  A model of an algebraic theory $(\Sigma_A, E_A)$ is a set $X$ together with an interpretation
  $\interpretation{\nop}{X} : X^n \rightarrow X$ of each $n$-ary operation $\nop$ in $\Sigma_A$,
  such that $\interpretation{s}{X} = \interpretation{t}{X}$ for each equation $s = t$ in $E_A$. Here
  $\interpretation{s}{X}$ is the interpretation of $s$ defined inductively using the interpretation of
  operations.

  A homomorphism between two models $(X, \interpretation{(-)}{X})$ and $(Y, \interpretation{(-)}{Y})$ is a morphism $h: X \rightarrow Y$ such that
  $h (\interpretation{\nop}{X}(x_1, \ldots, x_n)) = \interpretation{\nop}{Y}(h(x_1), \ldots, h(x_n))$
  for each $n$-ary operation $\nop$ in $\Sigma_A$ and $x_1, \ldots, x_n : X$.
  The models of an algebraic theory and homomorphisms between them form a category called the \emph{category of algebras} of the algebraic theory, denoted $A$-alg.
\end{defi}

The free model of an algebraic theory $A$ with variables in a set $X$ consists of the set of equivalence
classes of $A$-terms in context $X$.
The functor $F : \catset \rightarrow A$-alg sending each set $X$ to the free model of $A$ on $X$ is left adjoint to the forgetful functor sending each $A$-model to its underlying set. This adjunction induces a monad on $\catset$, called the \emph{free model monad} of the algebraic theory \cite{Linton1966,Lawvere1963,Manes1976}. The category of algebras of $A$ is isomorphic to the Eilenberg-Moore category of this monad.
If a monad $T$ is isomorphic to the free model monad of an algebraic theory, then we say that $T$ \emph{is presented} by that algebraic theory.

\begin{exa}[Boom Hierarchy Monads~\cite{Meertens1986}]
 \label{ex:boom-monads-alg}
  The binary tree monad, the list monad, the multiset monad and the powerset monad are the free model monads of respectively the theories of:
  \begin{itemize}
    \item Magmas, the theory consisting of a constant and a binary operation satisfying the left and right unit equations.
    \item Monoids, which are magmas satisfying the associativity equation.
    \item Commutative monoids.
    \item Idempotent commutative monoids, which are also known as semilattices.
  \end{itemize}
\end{exa}

\begin{exa}
The finite probability distributions monad is the free model monad of the theory of convex algebras \cite{Jacobs2010}.
\end{exa}

The equations of an algebraic theory determine much of the behaviour of its free model monad. For example, linear equations result in monads that always distribute over commutative monads \cite{ManesMulry2007, parlant-thesis}. In this paper, we built upon the ideas of Gautam \cite{gautam1957} and Parlant et al \cite{parlant2020}, and distinguish the following classes of equations:

\begin{defi}\label{def:equation-types}
    Write $\var(t)$ for the set of variables that appear in a term $t$. We say that an equation $s = t$ is:
    \begin{description}
      \item[Linear] if $\var(s) = \var(t)$ and each variable in these sets appears exactly once in both $s$ and $t$. Example: $x * y = y * x$.
      \item[Balanced] if $\var(s) = \var(t)$ and each variable in these sets appears equally many times in $s$ and $t$. Example: $(x * y) * (y * z) = (y * y) * (x * z)$.
      \item[Dup] if there is an $x : \var(s) \cup \var(t)$, such that $x$ appears $\geq 2$ times in $s$ and/or $t$. Example: the balanced equation above, as well as $x * x = x * x * x$ and $x \wedge (x \vee y) = x$.
      \item[Drop] if $\var(s) \neq \var(t)$. Example: $x \wedge (x \vee y) = x$.
    \end{description}
\end{defi}

\begin{rem}
  Notice that these types of equations are not mutually exclusive. An equation can for instance be both dup and drop, such as the absorption equation $x \wedge (x \vee y) = x$.
\end{rem}

\subsection{Distributive Laws}
One way of composing two monads is via a \emph{distributive law} describing the interaction between the two monads~\cite{Beck1969}.

\begin{defi}[Distributive Law]\label{def:distlaw}
Given monads $\langle S, \eta^S, \mu^S \rangle$ and $\langle T, \eta^T, \mu^T \rangle$, a distributive law distributing $S$ over $T$ is a natural transformation $\distlaw: ST \rightarrow TS$ satisfying the following axioms:
\begin{align}
  \distlaw \circ \eta^S T & = T\eta^S & \distlaw \circ S\eta^T  & = \eta^T S \tag{unit axioms} \\
  \distlaw \circ \mu^S T & = T\mu^S \circ \distlaw S \circ S\distlaw & \distlaw \circ S\mu^T  & = \mu^T S \circ T\distlaw \circ \distlaw T \tag{multiplication axioms}
\end{align}
\end{defi}

\begin{exa}[Lists and Multisets]
  Distributive laws are named after the well-known distributivity of multiplication over addition: $a * (b + c) = (a * b) + (a * c)$.
  Many distributive laws follow the same distribution pattern. For example, the list monad distributes over the multiset monad in this way: $\distlaw [\Lbag a \Rbag, \Lbag b, c \Rbag] = \Lbag [a,b], [a,c] \Rbag$.
  However, this is by no means the only way a distributive law can function.
\end{exa}

\begin{thmC}[\cite{Beck1969}]
  Let~$\mathcal{C}$ be a category, and~$\langle S, \eta^S, \mu^S \rangle$ and~$\langle T, \eta^T, \mu^T \rangle$ two monads on~$\mathcal{C}$. If~$\distlaw: ST \rightarrow TS$ is a distributive law, then the functor $TS$ carries a monad structure with unit~$\eta^T\eta^S$ and multiplication~$\mu^T\mu^S \circ T\distlaw S$.
\end{thmC}

We frequently use the following equivalence in our proofs:
\begin{thmC}[\cite{Beck1969}]\label{th:beck:dist-lift}
  Given two monads $\langle S, \eta^S, \mu^S \rangle$ and~$\langle T, \eta^T, \mu^T \rangle$ on a category~$\mathcal{C}$, there is a bijective
  correspondence between distributive laws of type $ST \rightarrow TS$, and liftings of $T$ to the Eilenberg-Moore category $\EM{S}$ of $S$.
  \end{thmC}
  Here, a lifting of $T$ to $\EM{S}$ is an assignment mapping $S$-algebra structures on a set $X$ to $S$-algebra structures on $TX$ such that $\eta^T$ and $\mu^T$ 
  are $S$-algebra homomorphisms, and such that $Tf$ is an $S$-homomorphism whenever $f$ is.

\section{Guarded Recursion and the Delay Monad}\label{sec:CCTT}

In this paper we work informally in Clocked Cubical Type Theory (CCTT)~\cite{CubicalCloTT}.
At present, this is the only known theory combining the features we need: Multiclocked guarded
recursion and quotient types (to express free monads).
Here we remind the reader of the basic principles of CCTT, but we refer to Kristensen et al.~\cite{CubicalCloTT}
for the full details, including a denotational semantics for CCTT.

\subsection{Algebraic Theories in Cubical Type Theory}

CCTT is an extension of Cubical Type Theory (CTT)~\cite{CTT}, which in turn is a version of Homotopy Type Theory (HoTT)~\cite{hottbook}
that gives computational content to the univalence axiom.
In CTT, the identity type of Martin-L{\"o}f type theory is replaced by a path type, which we shall write infix as $t \peq_A u$, often omitting the
type $A$ of $t$ and $u$. We will work informally with $\peq$, using its standard properties such as function extensionality.

A type $A$ is a \emph{homotopy proposition} (or hprop)
in HoTT and CTT, if any two elements of $A$ are equal, and an
\emph{hset} if $x \peq_A y$ is an hprop for all $x,y: A$. Assuming a universe of small types, one
can encode universes $\UhProp$ and $\UhSet$ of homotopy sets and propositions in
the standard way.
The benefit of working with hsets is that there is no higher structure
to consider. In particular, the collection of hsets and maps between these forms a category
in the sense of HoTT~\cite{hottbook}, and so basic category theoretic notions such as functors and
monads on hsets can be formulated in the standard way.

The notion of algebraic theory can also be read directly in CTT this way.
Moreover, the free monads on algebraic theories can be defined using higher inductive types
(HITs). These are types given inductively by constructors for terms as well as for equalities.
For example, the format for HITs used in CCTT~\cite{CubicalCloTT} (adapted
from Cavallo and Harper~\cite{CavalloHarper}) is expressive 
enough, as we show in Appendix~\ref{app:encoding:alg:th}.
We write type equivalence as $A \equi B$. For hsets this just means that there are maps
$f : A \to B$ and $g : B \to A$ that are inverses of each other, as
expressed in CTT using path equality.

\subsection{Multi-Clocked Guarded Recursion}

\begin{figure}
\begin{mathpar}
  \inferrule*
  {\wfcxt{\Gamma}}
  {\wfcxt{\Gamma, \kappa : \clocktype}{}}
  \and
  \inferrule*
  {\kappa : \clocktype \in \Gamma}
  {\wfcxt{\Gamma, \tickA : \kappa}{}}
  \and
  \inferrule*
  {\hastype{\Gamma, \timeless{\Gamma'}}{t}{\latbind{\tickA}{\kappa} A}\\ \wfcxt{\Gamma,\tickB:\kappa,\Gamma'}}
  {\hastype{\Gamma,\tickB: \kappa,\Gamma'}{\tapp[\tickB] t}{A\toksubst{\tickB}{\tickA}}}
  \and
  \inferrule*
  {\hastype{\Gamma,\tickA:\kappa}{t}{A}}
  {\hastype{\Gamma}{\tabs{\tickA}{\kappa} t}{\latbind{\tickA}{\kappa} A}}
  \and
    \inferrule*
  {\hastype{\Gamma,\kappa : \clocktype}{t}{A}}
  {\hastype{\Gamma}{\Lambda\kappa. t}{\forall \kappa . A}}
  \and
  \inferrule*
  {\hastype{\Gamma}{t}{\forall \kappa . A}\\
    \hastype\Gamma{\kappa'}\clocktype}
  {\hastype{\Gamma}{t [\kappa']}{A \subst{\kappa'}{\kappa}}}
   \and
  \inferrule*
  {\hastype{\Gamma}{t}{\later^\kappa A \to A}}
  {\hastype{\Gamma}{\dfix\,t}{\later^\kappa A}}
  \and
  \inferrule*
{\hastype[]{\Gamma}{t}{\later^\kappa A \to A} }
{\hastype{\Gamma}{\pfix[\kappa] \, t}{\latbind{\tickA}{\kappa}{\Path{A}{\tapp{(\dfix t)}}{t(\dfix t)}}}}
\end{mathpar}
\caption{Selected typing rules for Clocked Cubical Type Theory \cite{CubicalCloTT}.
  The telescope $\timeless{\Gamma'}$ is
  composed of the timeless assumptions in $\Gamma$, i.e. interval variables
  and faces (as in Cubical Type Theory) as well as clock variables.}
\label{fig:later:typing}
\end{figure}

CCTT extends CTT with multi-clock guarded recursion. The central component in this is a modal type operator $\later$ indexed by clocks
$\kappa$, used to classify data that is delayed by one time step on clock $\kappa$. The most important typing rules of CCTT are collected in
Figure~\ref {fig:later:typing}. Clocks are introduced as special assumptions $\kappa : \clocktype$ in a context, and can be abstracted and applied
to terms of the type $\forall \kappa . A$ which behaves much like a $\Pi$-type for clocks. Like function extensionality, extensionality for $\forall\kappa . A$ also holds in CCTT.

The rules for $\later$ also resemble those of $\Pi$-types: Introduction is by abstracting special assumptions $\tickA : \kappa$
called \emph{ticks} on the clock $\kappa$. Since ticks can appear in terms, the
modal type $\latbind\tickA\kappa A$ binds $\tickA$ in $A$, just like a $\Pi$-type binds a variable. We write $\later^\kappa A$ for
$\latbind\tickA\kappa A$ when $\tickA$ does not appear in $A$. The introduction rule for $\later$ can be
read as stating that if $t$ has type $A$ after the tick $\tickA$, then $\tabs{\tickA}{\kappa} t$ has type $\latbind{\tickA}{\kappa} A$ now.

The modality $\later$ is eliminated by applying a term to a tick. Note that the term $t$ applied to the tick $\beta$ cannot already contain $\tickB$
freely. This restriction prevents $t$ from being applied twice to the same tick, which would construct terms of type $\later^\kappa \later^\kappa A \to \later^\kappa A$, collapsing two steps into one. Moreover, $t$ cannot contain any variables nor other ticks occurring in the context after $\tickB$, only \emph{timeless}
assumptions, i.e., clocks, interval assumptions and faces. One application of timeless assumptions is to type the extensionality principle for $\later$:
\begin{equation} \label{eq:later:ext}
  (t \peq_{\latbindsmall{\tickA}{\kappa} A} u) \equi \latbind{\tickA}{\kappa} (\tapp t \peq_A \tapp u).
\end{equation}
For all explicit applications of terms to ticks in this paper, the term will not use timeless assumptions.
The usual $\eta$ and $\beta$ laws hold for tick abstraction and application.

As an applicative functor, $\later$ can even be given a dependent applicative action of type
\[
 \Pi(f : \later^\kappa(\Pi (x : A) . B(x)) . \Pi(y : \later^\kappa A) . \latbind\tickA\kappa B(\tapp y).
\]

Ticks are named in CCTT for reasons of normalisation~\cite{bahr2017clocks}, but are essentially identical. This is expressed in type theory as the \emph{tick irrelevance principle}:
\begin{equation} \label{eq:tirrAx}
  \tirrAx : \Pi(x : \later^\kappa A) . \,\latbind{\tickA}{\kappa}\,  \latbind{\tickB}{\kappa} (\Path{A}{\tapp{x}}{\tapp[\tickB] x}).
\end{equation}
The term $\tirrAx$ is defined in CCTT using special combinators on ticks, allowing for computational content to $\tirrAx$. This means that the rule
for tick application is more general than the one given in Figure~\ref {fig:later:typing}. However, we will not need this further generality for
anything apart from $\tirrAx$, which we use directly.

Finally, CCTT has a fixed point operator $\dfix[]$ which unfolds up to path equality as witnessed by $\pfix[]$. Using these, one can define
$\fix : (\later^\kappa A \to A) \to A$ as $\fix(t) = t(\dfix t)$ and prove $\fix(t) \peq t(\nextop^\kappa( \fix(t)))$ where
$\nextop^\kappa \defeq (\lambda (x : A). \tabs\tickA\kappa x) : A \to \later^\kappa A$.
Note that this relies on being able to introduce variables that appear before a tick in a context. This is not the case
in all Fitch-style modal type theories~\cite{clouston2018fitch,gratzer2020multimodal}.

\subsection{Guarded Recursive Types}

A guarded recursive type is a recursive type in which the recursive occurrences of the type are all guarded by a $\later$. These
can be encoded up to equivalence of types using fixed points of maps on the universe. Our primary example is the guarded recursive delay monad $\grD$
defined to map an $X$ to the recursive type
\begin{align*}
 \grD X \equi X + \later^\kappa(\grD X).
\end{align*}
We write 
\begin{align}\label{eq:now:step:kappa}
\now^\kappa  & : X \to \grD X & \step^\kappa & : \later^\kappa (\grD X) \to \grD X  
\end{align}
for the two maps given by inclusion and the equivalence above.

Since $\later^\kappa$ preserves the property of being an hset, one can prove by guarded recursion
that $\grD X$ is an hset whenever $X$ is. $\grD$
can be seen as a free construction in the following sense.

\begin{defi} \label{def:delay:algebra}
 A \emph{delay algebra} on the clock $\kappa$ is an hset $X$ together with a map
 $\later^\kappa X \to X$.
\end{defi}

Given an hset $X$, the hset $\grD X$ carries a delay algebra structure. It is the free delay algebra
in the sense that given any other delay algebra $(Y, \xi)$, and a map $f : X \to Y$,
there is a unique homomorphism $\overline f : \grD X \to Y$ extending $f$ along $\now$,
defined by the clause
\begin{align} \label{eq:extension:case:step}
 \overline f(\step^\kappa (x)) & = \xi(\tabs \tickA\kappa \overline f (\tapp x)).
\end{align}
This is a recursive definition that can be encoded as a fixed point of a map
$h : \later^\kappa(\grD X \to Y) \to (\grD X \to Y)$ defined using the clause
 $h(g)(\step^\kappa (x)) = \xi(\tabs \tickA\kappa (\tapp g) (\tapp x))$.
In this paper we use the simpler notation of (\ref{eq:extension:case:step}) for such definitions
rather than the explicit use of $\fix$.

We sketch the proof that $\overline f$ is the unique homomorphism extending $f$,
to illustrate the use of $\fix$ for proofs. Suppose $g$ is another such extension. To use
guarded recursion, assume that $\later^\kappa( g\peq \overline f)$. We show
that $g(\step^\kappa(x)) = \xi(\tabs \tickA\kappa \overline f (\tapp x))$. Since $g$ is a homomorphism:
 $g(\step^\kappa(x)) = \xi(\tabs \tickA\kappa g (\tapp x))$.
So by extensionality for $\later$ (\ref{eq:later:ext}) it suffices to show that
  $\latbind\tickA\kappa {(g (\tapp x) \peq \overline f (\tapp x))}$,
which follows from the guarded recursion hypothesis.

Tick irrelevance implies that $\grD$ is a commutative monad in the sense of Kock \cite{kock1970}.

\subsection{Encoding Coinductive Types}

Coinductive types can be encoded using a combination of guarded recursive types and
quantification over clocks. This was first observed by Atkey and McBride~\cite{atkey13icfp}.
We recall a special case of a more general theorem for this in
CCTT~\cite{CubicalCloTT}. First a definition.

\begin{defi}
 A functor $F : \UhSet \to \UhSet$ \emph{commutes with clock quantification}, if the canonical map
   $F(\forall \kappa . (\capp X)) \to \forall\kappa .F(\capp X)$
 is an equivalence for all $X : \forall\kappa. \UhSet$. An hset $X$ is \emph{clock irrelevant} if the
 constant functor to $X$ commutes with clock quantification, i.e. if the canonical map
 $X \to \forall\kappa . X$ is an equivalence.
\end{defi}

Note that functors commuting with clock quantification map clock irrelevant types to clock irrelevant types.

\begin{thmC}[\cite{CubicalCloTT}] \label{thm:encoding:coinduct}
 Let $F$ be an endofunctor on the category of hsets commuting with clock quantification, and
 let $\nu^\kappa F$ be the guarded recursive type satisfying
  $F(\later^\kappa(\nu^\kappa F)) \equi \nu^\kappa F$,
 then $\nu F \defeq \forall\kappa . \nu^\kappa F$ carries a final coalgebra structure for $F$.
\end{thmC}

In order to apply Theorem~\ref{thm:encoding:coinduct}, of course, one needs a large collection
of functors $F$ commuting with clock quantification. Fortunately, the collection of such functors
is closed under almost all type constructors, including finite sum and product, $\Pi$ and
$\Sigma$ types, $\later$, $\forall\kappa$, and guarded recursive types~\cite[Lemma~4.2]{CubicalCloTT}.
Clock irrelevant types are likewise closed under the same type constructors, and path equality. The only exception
to clock irrelevance is the universe type.

For example, if $X$ is clock irrelevant, then $F(Y) = X + Y$ commutes with clock quantification,
and so $\ciD X \defeq \forall\kappa . \grD X$ is the coinductive
solution to $\ciD X \equi X + \ciD X$. We write 
\begin{align*}
 \now & : X \to \ciD X & \step & : \ciD X \to \ciD X
\end{align*}
Note the distinction in notation from the corresponding inclusions for $\grD$ in (\ref{eq:now:step:kappa}).

CCTT moreover has a principle
of \emph{induction under clocks} allowing one to prove that many HITs are clock
irrelevant, including the empty type, booleans and natural numbers. Moreover, as the 
next proposition states, also free model monads commute with clock quantification.
The proof can be found in Appendix~\ref{app:encoding:alg:th}.

\begin{prop} \label{prop:free:monad:cirr}
 Let $A = (\Sigma_A, E_A)$ be an equational theory such that $\Sigma_A$ and $E_A$ are clock
 irrelevant. Then the free model monad $T$ commutes with clock quantification.
 In particular, $T(X)$ is clock irrelevant for all clock irrelevant $X$.
\end{prop}

The collection of clock irrelevant propositions can be shown to be closed under standard logical
connectives. Alternatively, one can assume a global clock constant $\clockc$, which then can
be used to prove that all propositions are clock irrelevant.

\begin{convention}
 In the remainder of this paper, the word \emph{set} will refer to a clock-irrelevant hset,
 and the word \emph{proposition} will refer to clock-irrelevant homotopy propositions. We
 will write $\USet$ and $\Prop$ for the universes of these. Similarly, whenever we mention functors these are assumed to commute with clock quantification.
\end{convention}

\section{Specific Combinations with Delay}\label{sec:examples}

In this section we look at some specific examples of monads, and see how they combine with the delay monads. In particular, we will look at the exception, reader, state, continuation, and selection monads. Intuitively, these monads model (parts of) the process: read input - compute - do something with the output. For instance, the state monad reads a state, then both updates the current state and gives an output. Combining the state monad with the delay monads allows us to model the fact that the computation in between reading the input and giving the output takes time, and might not terminate.

The examples we give follow the same pattern as the adaptation of these monads to domain theory: we insert a delay monad where one would use lifting
in the domain theoretic case. However, we also show that the algebraic status of these monads is much simpler in the guarded recursive case than
in the domain theoretic one: they can simply be understood as being generated by algebraic theories where one operation ($\step$) has a non-standard
arity. In the domain theoretic case, the algebraic description is in terms of enriched Lawvere theories~\cite{hyland2006discrete}.
We give no algebraic description of the combinations with the coinductive delay monads, because this does not by itself have an algebraic description.

First note that for combinations with delay via a distributive law, it is enough to find a distributive law for the guarded recursive version $\grD$.
\begin{lem} \label{lem:dist:grD:to:ciD}
 Let $T$ be a monad. A distributive law $\distlaw_X : \forall\kappa . T(\grD(X)) \to \grD (T(X))$ for the guarded delay monad induces a distributive law
 $T\ciD \to \ciD T$ for the coinductive delay monad. Similarly, if $T^\kappa$ is a family of monads indexed by $\kappa$ then $T(X) = \forall\kappa . T^\kappa(X)$ carries a monad structure.
\end{lem}

\begin{proof}
 The distributive law can be constructed as the composite
 \[
 \begin{tikzcd}
 T(\forall\kappa . \grD( X)) \ar{r}{} & (\forall\kappa . T(\grD(X))) \ar{r}{} & \forall\kappa. \grD(T(X)),
 \end{tikzcd}
 \]
 of the canonical map (mapping $t$ to $\Lambda\kappa . T(\mathsf{ev}_\kappa)\,t$, where $\mathsf{ev}_\kappa \,x =\capp x $) and $\forall\kappa. \capp\distlaw$. 
  
 We show the proof for one of the unit axioms, the other unit axiom and the multiplication axioms follow similarly.
 \[
 \begin{tikzcd}
  & \forall\kappa . \grD( X)  \ar[ld, "\eta^T"']\ar{d}{\forall\kappa . \eta^T}\ar{rd}{\forall\kappa . \grD( \eta^T)} &  \\
  T(\forall\kappa . \grD( X)) \ar{r} & (\forall\kappa . T(\grD(X))) \ar{r}{\forall\kappa. \capp\distlaw} & \forall\kappa. \grD(T(X))
 \end{tikzcd}
 \]
 Here, the left triangle commutes due to naturality of $\eta^T$, and the right triangle commutes because $\distlaw$ is a distributive law. 
 
 For the second statement, define the multiplication as the composite
 \[\forall\kappa . T^\kappa(\forall\kappa' . T^{\kappa'}(X)) \to (\forall\kappa . T^\kappa(T^\kappa(X))) \to \forall\kappa . T^\kappa(X).\]
 We show that this satisfies the multiplication axiom of a monad, that is, we show that the outer square below commutes. This follows from the fact that all four inner squares commute: The top left square commutes because both paths set $\kappa'$ and $\kappa''$ to $\kappa$. The top right square commutes by naturality of $\mu^\kappa$. The lower left square commutes because applying the parameterised multiplication $\mu^{\kappa'}$ and then instantiating it for $\kappa' = \kappa$ is the same as first choosing $\kappa' = \kappa$ and then applying $\mu^\kappa$. And finally, the square in the lower right corner is the multiplication axiom for $T^\kappa$, which holds because we assume that each $T^\kappa$ is a monad. 
 \begin{center}
   \begin{tikzcd}
      \forall\kappa . T^\kappa(\forall\kappa' . T^{\kappa'}( \forall\kappa'' . T^{\kappa''}(X) )) \ar[d] \ar[r] & \forall\kappa . T^\kappa(T^\kappa( \forall\kappa'' . T^{\kappa''}(X))) \ar[r, "\forall\kappa . \mu^\kappa"]\ar[d] & \forall\kappa . T^\kappa( \forall\kappa'' . T^{\kappa''}(X)) \ar[d] \\
      \forall\kappa . T^\kappa(\forall\kappa' . T^{\kappa'}(T^{\kappa'}(X)) ) \ar[d, "T(\forall\kappa'. \mu^{\kappa'})"]\ar[r] &
      \forall\kappa . T^\kappa(T^{\kappa}(T^{\kappa}(X))) \ar[r, "\forall\kappa . \mu^\kappa"]\ar[d, "T \mu^\kappa" ]
       & \forall\kappa . T^\kappa(T^\kappa(X)) \ar[d, "\forall\kappa . \mu^\kappa"] \\
      \forall\kappa . T^\kappa(\forall\kappa' . T^{\kappa'}(X)) \ar[r] &  \forall\kappa . T^\kappa(T^\kappa(X)) \ar[r, "\forall\kappa . \mu^\kappa"] & \forall\kappa . T^\kappa(X)
   \end{tikzcd}
 \end{center}
 The proof showing that $T(X) = \forall\kappa . T^\kappa(X)$ also satisfies the unit axioms of a monad is similar.
\end{proof}

\paragraph*{Exceptions}
The first monad we consider is the exception monad. For a set of exceptions $E$, the exception monad is given by the functor $(-+E)$, with obvious unit and multiplication. The exception monad is the free model monad of the algebraic theory consisting of a signature with a constant $e$ for each exception in $E$, and no equations.

It is well known that the exception monad distributes over any monad, and therefore we have a distributive law
$\distlaw : (\grD(-)+E)  \rightarrow \grD (-+E)$.
The resulting composite monad $\grD (-+E)$ is the free model monad of the theory consisting of constants $e : E$ and a $\step$-operator forming a delay algebra, with no additional equations.

\paragraph*{Reading}
The reader monad $(-)^R$ is presented by the algebraic theory consisting of a single operation $\lookup : X^R \rightarrow X$, satisfying the equations
\begin{align*}
\forall x : X . \, & \lookup (\lambda r. x) = x &
\forall g : (X^R)^R . \, & \lookup (\lookup \circ g) = \lookup (\lambda s . g \, s\, s).
\end{align*}

To combine the reader monad with the delay monad, we define a distributive law $\grD R \rightarrow R\grD$ by the clauses
\begin{align*}
\distlaw_{R\grD} (\now^\kappa f) & = \lambda r. \now^\kappa (f r) \\
\distlaw_{R\grD} (\step^\kappa d) & = \lambda r. \step^\kappa (\tabs{\tickA}{\kappa} (\distlaw_{R\grD} (\tapp[\tickA]{d})) r),
\end{align*}
where $f : X^R$ and $d : \later^\kappa (X^R)$. The resulting composite monad $R\grD$ is the free model monad of the theory consisting of $\lookup$ and $\step$ satisfying the above equations for $\lookup$, and
\begin{align} \label{eq:dist:lookup:step}
\forall d : \later^\kappa (X^R) . \, \step^\kappa (\tabs{\tickA}{\kappa} \lookup (\tapp[\tickA]{d})) &  = \lookup (\lambda r . \step^\kappa (\tabs{\tickA}{\kappa} \tapp[\tickA]{d} \,r)).
\end{align}
Diagrammatically:
\begin{center}
  \begin{tikzcd}
    \later^\kappa (X^R) \ar{d}{\later^\kappa (\lookup)} \ar{r}{} & (\later^\kappa X)^R \ar{r}{{\step{^\kappa}}^R} & X^R \ar{d}{\lookup} \\
    \later^\kappa X \ar{rr}{\step^\kappa} & & X 
  \end{tikzcd}
\end{center}

\paragraph*{Writing}
For a monoid $M$, the writer monad forms the cartesian product $M \times (-)$, again with obvious unit and multiplication. Its algebras are sets with a left $M$-action. The writer monad distributes over the delay monad via a distributive law of type $(M\times\grD(-)) \rightarrow \grD(M\times(-))$:
\begin{align*}
\distlaw_{\grD W} (m , \now^\kappa x) & = \now^\kappa (m , x) \\
\distlaw_{\grD W} (m , \step^\kappa d) & = \step^\kappa (\tabs{\tickA}{\kappa} (m , \distlaw (\tapp[\tickA]{d}))).
\end{align*}
This has the effect of delaying all output until the computation is finished. 

There is also a distributive law combining the reader and writer monads, of type $(M\times(-)^R) \rightarrow (M\times(-))^R$, given by:
\begin{align*}
\distlaw_{RW} (m , f) & = \lambda r. (m , (f r)).
\end{align*}
When $M = R$, this functor is the same as that of the state monad. However, it has a different monad structure: The unit of this monad is the composite of the units of the reader and writer monads, sending an $x : X$ to $\lambda r . (\varepsilon , x)$, where $\varepsilon$ is the unit of the monoid $M$. Contrast this with the unit of the state monad, which sends an $x : X$ to $\lambda r . (r , x)$. The former changes the incoming state $r$ to some special state $\varepsilon$, whereas the latter leaves the incoming state $r$ unchanged. 

Still, it is interesting to note that the three distributive laws $\distlaw_{R\grD}, \distlaw_{\grD W}$, and $\distlaw_{RW}$ satisfy the Yang-Baxter equation, meaning that they give rise to a monad structure on the functor $(\grD(M \times (-)))^R$ \cite{Cheng2011}. This places the delay monad in between the reading and writing steps, which is exactly what we will do in our combination of the delay monad with the actual state monad.

\paragraph*{Global State}
Plotkin and Power~\cite{plotkin2002} show that the global state monad $(S \times -)^S$ can be described algebraically by two operations:
$\lookup : \,  X^S \rightarrow X$ and $\update : \,  X \rightarrow X^S$, satisfying four equations:
\begin{align*}
 &\forall f : (X^S)^S .\, \lookup (\lambda s . \lookup (f s)) = \lookup (\lambda s . f s s) \\
 &\forall x : X .\,  \lookup (\update x) = x \\
 &\forall g : (X^S) .\,   \forall s : S .\, \update (\lookup g) s = \update (g s) s \\
 &\forall x : X .\, \forall s t : S .\, \update ((\update x) s) t = (\update x) s
\end{align*}
They call the category of such algebras \emph{GS-algebras}.

The natural combination of global state and $\grD$ is $(\grD(S \times -))^S$, describing computations whose steps occur between reading the initial state and writing back the updated state. This combination does not arise via a distributive law, but rather via the well-known monad transformer for the state monad \cite{benton2000}.
 
To describe this monad algebraically, define a GSD-algebra to be a GS-algebra that also carries a delay algebra structure, which interacts with the GS-algebra structure as follows:
\begin{align*}\label{eq:lookup-update-step-interact}
\forall d : \later^\kappa (X^S) . \,  & \step^\kappa (\tabs{\tickA}{\kappa} \lookup (\tapp[\tickA]{d}))   = \lookup (\lambda s . \step^\kappa (\tabs{\tickA}{\kappa} \tapp[\tickA]{d} \,s)) \\
\forall x : \later^\kappa X . \,  & \lambda s . \update (\step^\kappa x) s
 = \lambda s . \step^\kappa (\tabs{\tickA}{\kappa} \update (\tapp[\tickA]{x}) s).
\end{align*}
Diagrammatically:
\begin{center}
  \begin{tikzcd}
    \later^\kappa (X^S) \ar{d}{\later^\kappa (\lookup)} \ar{r}{} & (\later^\kappa X)^S \ar{r}{{\step^\kappa}^S} & X^S \ar{d}{\lookup}  &
        \later^\kappa X \ar{rr}{\later^\kappa (\update)}\ar{d}{\step^\kappa} && \later^\kappa (X^S)\ar{r}{} & (\later^\kappa X)^S \ar{d}{{\step^\kappa}^S} \\
    \later^\kappa X \ar{rr}{\step^\kappa} & & X &
     X \ar{rrr}{\update} & & & X^S
  \end{tikzcd}
\end{center}

\begin{thm} \label{thm:state:free}
The monad $(\grD(S \times -))^S$ is the free model monad of the theory of GSD-algebras.
\end{thm}

\begin{proof}
We show that $(\grD(S \times -))^S$ maps a set $X$ to the carrier set of the free GSD-algebra.
First note that $(\grD(S \times X))^S$ carries a GSD-algebra structure $(\lookupspec, \updatespec, \stepspec)$ defined as follows:
\begin{align*}
\forall x : ((\grD(S \times X))^S)^S, s : S . \, & \lookupspec \,x\, s = x\, s\, s \\
\forall x : (\grD(S \times X))^S . \, & \updatespec\, x = \lambda s . \lambda s' . x \, s \\
\forall x : \later^\kappa ((\grD(S \times X))^S) . \, & \stepspec \,x = \lambda s . \step^\kappa(\tabs{\tickA}{\kappa} (\tapp[\tickA]{x}) \,s).
\end{align*}
and that there is an inclusion $\eta^{DS} : X \to (\grD(S \times X))^S$ defined as $\eta^{DS}(x) = \lambda s . \now^\kappa (s, x)$. 

Given a GSD-algebra $Y$, and any map $f : X \rightarrow Y$, 
we must show that there is a unique homomorphism $\bar{f} : (\grD(S \times X))^S \rightarrow Y$ 
such that $\bar{f} \circ \eta^{DS} = f$. Define 
\[
\bar{f}(x) = \lookup_Y (\lambda s . f' (x \,s)),
\]
where $f' : \grD(S \times -) \rightarrow Y$ is defined as:
\begin{align*}
f'(\now^\kappa (s, x)) = \; & \update_Y (f(x))\, s \\
f'(\step^\kappa d) = \; & \step_Y(\tabs{\tickA}{\kappa} f' (\tapp d)).
\end{align*}
We omit the straightforward verification that $\bar f$ is a homomorphism and that $\bar f \circ\eta^{DS} = f$.

Finally, we show that $\bar{f}$ is the unique extension of $f$. Suppose that $g$ is a GSD-algebra homomorphism and that $g \circ \eta^{DS} = f$. We must show that $g = \bar{f}$.

First notice that, since $g$ is a homomorphism:
\begin{align*}
g(x) & = g(\lookupspec (\updatespec x)) \\
     & = g(\lookupspec (\lambda s . \lambda s' . x \,s )) \\
     & =  \lookup_Y (\lambda s . g(\lambda s' . x \, s)).
\end{align*}

So it suffices to show that $g(\lambda s' . x) = f'(x)$, for all $x : \grD(S \times X)$. We prove this by guarded recursion and cases of $x$. 

Suppose first that $x = \now^\kappa(x', s)$, then
\begin{align*}
  f'(x) & = \update_Y (f(x')) \,s \\
  & = \update_Y (g \circ \eta^{DS}(x')) s \\
  & = \update_Y (g (\lambda t . \now^\kappa (t, x'))) s \\
  & = g (\updatespec (\lambda t . \now^\kappa (t, x')) s) \\
  & = g(\lambda s' . \now^\kappa (s, x')) \\
  & = g(\lambda s' . x).
\end{align*}

If $x = \step^\kappa d$, then
\begin{align*}
f'(x) & = \step_Y(\tabs{\tickA}{\kappa} f' (\tapp d)) \\
 & = \step_Y(\tabs{\tickA}{\kappa} g (\lambda s . \tapp d)) \\
% & = \step_Y(\tabs{\tickA}{\kappa} g (\tabs{\tickB}{\kappa} \tapp[\tickA]{\lambda s' . \tapp d})) \\
 & = \step_Y((\later^\kappa g) (\tabs{\tickA}{\kappa}{ \lambda s . \tapp[\tickA] d})) \\
 & = g(\stepspec(\tabs{\tickA}{\kappa} \lambda s . \tapp d )) \\
 & = g(\lambda s . \step^\kappa(\tabs{\tickA}{\kappa} \tapp d)) \\
 & = g(\lambda s . x).
\end{align*}

So $\bar{f}$ is the unique extension of $f$, proving that the monad $(\grD(S \times -))^S$ is the free model monad of the theory of GSD-algebras.
\end{proof}

\noindent Note that also $(\ciD(S \times -))^S$ is a monad by Lemma~\ref{lem:dist:grD:to:ciD}, since $(\ciD(S \times -))^S \equi \forall\kappa.((\grD(S \times -))^S)$.

\paragraph*{Continuations}
The continuation monad $\mathcal{C}X = (X \rightarrow R) \rightarrow R$ is a classic example of a monad that is not algebraic. Indeed, Hyland et al. mention that continuations have more of a \emph{logical} character than an algebraic one \cite{hyland2007}. Still, it is a monad often used in combination with algebraic monads. In op. cit. they show that algebraic operations on the continuations monad are in 1-1 correspondence with algebraic operations on $R$, where the former can be constructed pointwise from the latter.

While $\step^\kappa$ is not strictly an algebraic operation due to its non-standard arity, it does follow this pointwise construction. 
\begin{prop}
The continuation monad $\mathcal{C}X$ has a delay algebra structure for all $X$ iff $R$ has a delay algebra structure. 
\end{prop}
\begin{proof}
If $\step^\kappa_R : \later R \rightarrow R$, define $\step^\kappa_{\mathcal{C}X}: \later ((X \rightarrow R) \rightarrow R) \rightarrow (X \rightarrow R) \rightarrow R$ as 
\[
\step^\kappa_{\mathcal{C}X} d \,g = \step^\kappa_R (\tabs{\tickA}{\kappa} (\tapp[\tickA]{d} (g)))
\]
Conversely, given $\step^\kappa_{\mathcal{C}X}: \later ((X \rightarrow R) \rightarrow R) \rightarrow (X \rightarrow R) \rightarrow R$ for any $X$, 
and $d : \later R$, instantiate $\step^\kappa_{\mathcal{C}X}$ with $R$, the constant map to $d$, and the identity on $R$ to get an element of $R$:
\[
\step^\kappa_R d = (\step^\kappa_{\mathcal{C}R} (\tabs{\tickA}{\kappa} \lambda g. (\tapp[\tickA]{d}))) (\id_R) \qedhere
\]
\end{proof}

Other than lifting operations from $R$ to $\mathcal{C}X$, Hyland et al. investigate sums, tensors and monad transformers as ways to combine algebraic monads with the continuation monad \cite{hyland2007}. As we show below, the sum of $\grD$ with any other monad always exists, so in particular the sum $\monadsum {\mathcal{C}}{\grD}$. The tensor product of $\mathcal{C}$ with the coinductive delay monad {\ciD} exists in $\catset$, since $\ciD$ has countable rank, and so Theorem 6 of \cite{hyland2007} applies, 
at least for $\ciD$ considered an endofunctor on sets in the classical sense. We have not worked out the details in the setting of Clocked Cubical Type Theory that we work in here. 
The monad transformer \cite{benton2000} combines the guarded delay monad with the continuations monad by wrapping the delay monad around $R$, and then applying the above lifting of delay-algebras on $\grD R$ to delay algebras on $(X \rightarrow (\grD R)) \rightarrow (\grD R)$.

Combinations via a distributive law are absent in this analysis of Hyland et al. A combination of the continuation monad with the delay monad via a distributive law would, depending on the direction of the distributive law, result in either a monad of type $\grD((X \rightarrow R) \rightarrow R)$ or $((\grD X) \rightarrow R) \rightarrow R$. However, we believe that neither distributive law exists. We will prove that the former case is not possible in Corollary~\ref{cor:no:go:continuation} below, if $R$ has a binary operation that is idempotent and commutative (for instance, if $R$ is the set of Booleans).  
For the latter, to define a distributive law, we would need to construct a continuation out of an element of form $\step^\kappa d$, where we cannot use guarded recursion to access $d$. Without any information about $d$, it seems unlikely that we can define something meaningful here.

Distributive laws involving the coinductive delay monad and the continuations monad might still exist, but have so far eluded us.

\paragraph*{Selections}

The selection monad $\mathcal{J}X = (X \rightarrow R) \rightarrow X$ \cite{escardo_oliva_2010} is a close companion to the continuation monad $\mathcal{C}X = (X \rightarrow R) \rightarrow R$, with many applications in game theory and functional programming \cite{escardo2010b, Hedges2014}. It takes a function $X \rightarrow R$, and selects an input $x : X$ to return. This could, for example, be an input for which the function attains an optimal value. 

Contrary to the continuation monad, the selection monad easily combines with the delay monad via a distributive law of type $\grD \mathcal{J} \rightarrow \mathcal{J}\grD$. 
It is similar to the distributive law for the delay monad over the reader monad, and is given by:
\begin{align*}
\distlaw (\now f) & = \lambda g. \now f(\lambda x. g(\now x)) \\
\distlaw (\step d) & = \lambda g. \step (\tabs{\tickA}{\kappa} (\distlaw (\tapp[\tickA]{d})) g),
\end{align*}
where $f: (X \rightarrow R) \rightarrow X$, $g : (\grD X) \rightarrow R$ and $d : \later^\kappa (\grD((X \rightarrow R) \rightarrow X))$.
Compared to the distributive law for the reader monad, the case for $\now f$ is slightly more complex, because $f$ takes a function from $X$ to $R$, but $g$ has type $\grD X \rightarrow R$. This is solved by giving $f$ the function $\lambda x. g(\now x)$, which is indeed a function from $X$ to $R$.

In fact, Fiore proved that there is a distributive law of type $T \mathcal{J} \rightarrow \mathcal{J} T$ for any strong monad $T$ \cite{FioreNotes}. 
The distributive law for the delay monad over the selection monad given above is an instance of this general distributive law.

Rather than using distributive laws to combine the selection monad with other monads, Escard{\'o} and Oliva study \emph{$T$-selection functions}: $\mathcal{J}^T X = (X \rightarrow R) \rightarrow TX$ for any strong monad $T$ \cite{escardo2017}, with the extra requirement that $R$ is a $T$-algebra. Where a distributive law wraps the monad $T$ around both occurrences of $X$ in the selection monad, $T$-selection functions only wrap $T$ around the resulting $X$.
Since $\grD$ is a strong monad, we also get the monad $\mathcal{J}^\grD$.

It is tempting to believe that, under the assumption that $R$ has a $T$-algebra structure, $\mathcal{J}^T X$ is a retract of $\mathcal{J}T X$. 
Indeed, we can define maps $\phi : \mathcal{J}^T \rightarrow \mathcal{J}T$ and $\psi : \mathcal{J}T \rightarrow \mathcal{J}^T$ such that $\psi \circ \phi = \id$.
  \begin{align*}
    \phi(f) & = \lambda g . f (\lambda x . g (\eta^T x))\\
    \psi(f') & = \lambda g' . f' (\lambda t . (\alpha_R \circ T(g)) t).
  \end{align*}
However, these maps are not monad maps, as they fail to commute with the multiplications of the two monads.

\paragraph*{Free Combinations With Delay}
\label{sec:free}

The sum of two monads $T$ and $S$ is a monad $\monadsum TS$ whose algebras are objects $X$ with algebra structures for both $T$ and $S$~\cite{hyland2006}. In terms of algebraic theories, the sum can be understood as combining two theories with no equations between them.
The sum of $\grD$ with any other monad always exists \cite[Theorem 4]{hyland2006}:

\begin{prop} \label{thm:free:comb}
  Let $T$ be a monad, and define $\monadsum T{\grD}$ as the guarded recursive type:
  \[
  (\monadsum T{\grD}) X \equi T(X + \later^\kappa ((\monadsum T{\grD}) X)).
  \]
  Then $(\monadsum T{\grD}) (X)$ is the carrier of the free $T$-algebra and delay algebra structure.
\end{prop}
\begin{proof}
It is easy to see that $(\monadsum T{\grD}) (X)$ is both a $T$-algebra, and a delay algebra. To see that it is the free one, suppose that $Y$ is also a $T$-algebra and a delay algebra, given by $\alpha_Y : TY \rightarrow Y$ and $\step^\kappa_Y : \later^\kappa Y \rightarrow Y$ respectively, and suppose that we have a map $f : X \rightarrow Y$. We define $\bar{f} : (\monadsum T{\grD}) (X) \rightarrow Y$ as:
\begin{align*}
  \bar{f} & = \alpha_Y (T(g)), 
\end{align*}
where 
\begin{align*}
g (\inl x) & = f x & 
  g (\inr x) & = \step^\kappa_Y(\tabs{\tickA}{\kappa} \bar{f} (\tapp[\tickA]{x}))
\end{align*}
Uniqueness follows by guarded recursion (similar to the proof underneath Equation~(\ref{eq:extension:case:step})), and from the fact that $TX$ is the free $T$-algebra.
\end{proof}

The monad mapping $X$ to $\forall\kappa . (\monadsum T{\grD}) (X)$ includes the coinductive delay monad and $T$, but we have not been
able to prove a general universal property for this. We believe that it is not the sum of the two.

\section{Parallel and Sequential Distribution of Operations}
\label{sec:parseq}

We now consider distributive laws of type $TD \rightarrow DT$, where $D$ is one of the delay monads and $T$ is any presentable monad.
Such laws equip the composite $DT$ with a monad structure, which is the natural one in particular for monads describing data structures, such
as those in the Boom hierarchy.

We again focus on distributive laws involving the guarded version of the delay monad, invoking Lemma \ref{lem:dist:grD:to:ciD}. Intuitively, such a distributive law pulls all the steps out of the algebraic structure of $T$: it turns a $T$-structure with delayed elements into a delayed $T$-structure.
There are two obvious candidates for such a lifting: \emph{parallel} and \emph{sequential} computation.
We define both of these on operations using guarded recursion. A lifting of terms then follows inductively from lifting each operation in the signature of the presentation of $T$.

\begin{defi}[Parallel Lifting of Operators]
  \label{parlift}
Let $A$ be an algebraic theory, and let $X$ be an $A$-model. Define, for each n-ary operation
$\binop$ in $A$, a lifting $\parinterpretation{\binop}{\grD X} : (\grD X)^n \to \grD X$ by:
\begin{align*}
    \parinterpretation{\binop}{\grD X} (\now^\kappa x_1, \dots,  \now^\kappa x_n) & = \now^\kappa (\interpretation{\binop}{X}(x_1, \dots, x_n)) \\
    \parinterpretation{\binop}{\grD X} (x_1, \dots, x_n) & = \step^\kappa(\tabsshort\tickA\kappa (\parinterpretation{\binop}{\grD X} (x_1', \dots, x_n'))),
\end{align*}
where the second clause only applies if one of the $x_i$ is of the form $\step(x_i'')$ and
\[
 x_i' =
\begin{cases}
 x_i & \text{ if } x_i = \now^\kappa(x_i'') \\
 \tapp{x_i''} & \text{ if } x_i = \step^\kappa(x_i'')
\end{cases}
\]
\end{defi}

\begin{defi}[Sequential Lifting of Operators]
  \label{seqlift}
Let $A$ be an algebraic theory, and let $X$ be an $A$-model.
Define, for each n-ary operation $\binop$ in $A$, a lifting $\seqinterpretation{\binop}{\grD X} : (\grD X)^n \to \grD X$ by:
\begin{align*}
    \seqinterpretation{\binop}{\grD X} (\now^\kappa x_1, \dots, \now^\kappa x_n) & = \now^\kappa (\interpretation{\binop}{X}(x_1, \dots,  x_n)) \\
    \seqinterpretation{\binop}{\grD X} (\now^\kappa x_1, \dots, \step^\kappa x_i, \dots x_n)
    & = \step^\kappa(\tabsshort\tickA\kappa(\seqinterpretation{\binop}{\grD X} (\now^\kappa x_1, \dots, (\tapp{x_i}), \dots, x_n))), %\\
\end{align*}
where, in the second clause, the $i$th argument is the first not of the form $\now^\kappa(x_k')$.
\end{defi}

In general, for an $n$-ary operation $\binop$, parallel lifting evaluates all arguments of the form
$\step (x_i)$ in parallel, and sequential lifting evaluates them one by one from the left.
Parallel lifting of an operator therefore terminates in as many steps as the maximum required
for each of its inputs to terminate, while sequential lifting terminates in the sum of the number of steps required for each
input.

The evaluation order of arguments in the case of sequential lifting is inessential, which can be proved using guarded recursion and tick irrelevance.

\begin{lem}\label{lem:takestepsoutanyorder-binary}
Let $A$ be an algebraic theory, and let $\nop$ be an $n$-ary
operation in $A$. Then
  \begin{align*}
   \seqinterpretation{\binop}{\grD X}(x_1, \ldots, \,\step^\kappa(x_i), \dots x_n) =
  & \step^\kappa(\tabs\tickA\kappa \seqinterpretation{\binop}{\grD X} (x_1, \ldots, \tapp{x_i}, \ldots, x_n)).
  \end{align*}
\end{lem}

\begin{proof}
  For readability, we prove this just in the case of $n=2$ and $i=2$, naming the variables $x$ and $y$.
  If $x = \now^\kappa x'$, then the equation holds by definition of $\seqinterpretation{\binop}{\grD X}$.
  If $x = \step^\kappa x'$, then:
    \begin{align*}
        \seqinterpretation{\binop}{\grD X}(x, \step^\kappa y)
      = \; & \step^\kappa(\tabsshort\tickA\kappa \seqinterpretation{\binop}{\grD X}((\tapp {x'}), \step^\kappa y)) \\
      = \; & \step^\kappa(\tabsshort\tickA\kappa \step^\kappa(\tabsshort\tickB\kappa (\seqinterpretation{\binop}{\grD X} ((\tapp {x'}), (\tapp[\tickB]{y}))))) \\
      = \; & \step^\kappa(\tabsshort\tickA\kappa \step^\kappa(\tabsshort\tickB\kappa (\seqinterpretation{\binop}{\grD X} ((\tapp[\tickB]{x'}), (\tapp y))))) \\
      = \; & \step^\kappa(\tabsshort\tickA\kappa  \seqinterpretation{\binop}{\grD X} (\step^\kappa x', (\tapp y))) \\
      = \; & \step^\kappa(\tabsshort\tickA\kappa  \seqinterpretation{\binop}{\grD X} (x, (\tapp y))),
    \end{align*}
    using guarded recursion and tick irrelevance.
\end{proof}

\subsection{Preservation of Equations}

Parallel lifting preserves all non-drop equations, whereas sequential lifting only preserves balanced equations. We prove this in the two following propositions. We write $\parinterpretation{s}{\grD X}$ for the interpretation of a term $s$ on $\grD X$ defined by induction of $s$ using the parallel lifting of operations, and likewise $\seqinterpretation{s}{\grD X}$ for the interpretation defined using sequential lifting of operations.

\begin{prop}[Parallel Preserves Non-Drop]\label{prop:par-eq-preservation}
  Let $A = (\Sigma_A, E_A$) be an algebraic theory, $X$ an $A$-model, and
  $s = t$ a non-drop equation that is valid in $A$.
  Then also $\parinterpretation{s}{\grD X} = \parinterpretation{t}{\grD X}$.
\end{prop}

\begin{proof}
Let $n$ be the number of free variables of $s$ and $t$ (by the non-drop assumption $s$ and $t$ have the same free variables). Then $\parinterpretation{s}{\grD X},
\parinterpretation{t}{\grD X} : (\grD X)^n \to \grD X$,
and an easy induction shows that these both satisfy the parallel lifting equations of Definition~\ref{parlift}. 
The case where all $x_i$ are of the form $\now(x_i')$ is easy. So suppose that at least one of the arguments
$x_1, \dots, x_n$ is a step and write
\[
 x_i' =
\begin{cases}
 x_i & \text{ if } x_i = \now^\kappa(x_i'') \\
 \tapp{x_i''} & \text{ if } x_i = \step^\kappa(x_i'')
\end{cases}
\]
Then
\begin{align*}
  \parinterpretation{s}{\grD X} (x_1, \dots, x_n) & = \step^\kappa(\tabs\tickA\kappa (\parinterpretation{s}{\grD X}(x_1', \dots, x_n'))) \\
  & = \step^\kappa(\tabs\tickA\kappa (\parinterpretation{t}{\grD X}(x_1', \dots, x_n'))) \\
  & = \parinterpretation{t}{\grD X} (x_1, \dots, x_n)
\end{align*}
by guarded recursion.
\end{proof}

The restriction to non-drop equations is necessary, because divergence in a dropped variable leads to divergence on one side of the equation, but not on the other.

M{\o}gelberg and Vezzosi \cite{mogelberg2021two} observed that parallel lifting does not define a distributive law in the case of the
finite powerset monad. Their proof uses idempotency, but in fact parallel lifting does not define a monad even just in the presence of a
single binary operation.

\begin{thm}\label{thm:mogelvezzo}
 Let $T$ be an algebraic monad with a binary operation $\binop$. Then the natural transformation $\distlaw : T\ciD \rightarrow \ciD T$ induced by parallel lifting does not define a distributive law, because it fails the second multiplication axiom.
\end{thm}
\begin{proof}
The counter example is the same as used by M{\o}gelberg and Vezzosi:
\begin{align*}
\parinterpretation{\nop}{\ciD \ciD X}(\mu^\ciD (\now (\step \now x)), \mu^\ciD (\step (\now (\now y))) & =
\step (\now (\interpretation{\nop}{X}(x, y))) \\
\mu^\ciD  (\parinterpretation{\nop}{\ciD \ciD X} (\now (\step (\now x)), \step (\now (\now y))))
& =\step (\step (\now (\interpretation{\nop}{X}(x, y)))). \qedhere
\end{align*}
\end{proof}
Note that we used the coinductive version of the delay monad in the above theorem. By Lemma \ref{lem:dist:grD:to:ciD}, this implies the same result for the guarded recursive version.

\begin{prop}[Sequential Preserves Balanced]\label{prop:seq-eq-preservation}
  Let $A$ be an algebraic theory, and $s, t$ be two $A$-terms such that $s = t$ is a balanced equation that is valid in $A$. Then also $\seqinterpretation{s}{\grD X} = \seqinterpretation{t}{\grD X}$.
\end{prop}

\begin{proof}
Suppose $s$ and $t$ both have $n$ free variables. Then:
\begin{align*}
 \seqinterpretation{s}{\grD X} (\now^\kappa(x_1), \dots, \now^\kappa(x_n)))
 & = \now^\kappa(\interpretation{s}{X}(x_1, \dots, x_n)) \\
 & = \now^\kappa(\interpretation{t}{X}(x_1, \dots, x_n)) \\
 & = \seqinterpretation{s}{\grD X} (\now^\kappa(x_1), \dots, \now^\kappa(x_n))).
\end{align*}
Suppose one of the inputs steps, say $x_i = \step^\kappa(x_i')$, and that
$x_i$ occurs $k$ times in both $s$ and $t$. Then an easy induction using
Lemma~\ref{lem:takestepsoutanyorder-binary} and tick irrelevance
shows that
\begin{align*}
\seqinterpretation{s}{\grD X}(x_1,  \dots, \step^\kappa(x_i'), \dots, x_n) & =
  \step^\kappa (\tabs\tickA\kappa((\nextstep)^{k-1} (\seqinterpretation{s}{\grD X}(x_1, \dots, \tapp{x_i'}, \dots, x_n)))),
\end{align*}
where $\nextstep = \step^\kappa \circ \nextop^\kappa$,
and likewise for $\seqinterpretation{t}{\grD X}(x_1,  \dots, \step(x_i'), \dots, x_n)$. The result now follows from
guarded recursion.
\end{proof}

Balance is necessary. For example, if $t$ and $s$ are terms in a single variable which occurs twice in $t$ and once in $s$, then
$\seqinterpretation{t}{\grD X}(\step^\kappa(x))$ takes at least two steps, but $\seqinterpretation{s}{\grD X}(\step^\kappa(x))$ might take only one.
Building on Proposition~\ref{prop:seq-eq-preservation}, one can prove the following.

\begin{thm}\label{thm:seq-lifting-distlaw}
Let $T$ be the free model monad of algebraic theory $\bb{T} = (\Sigma_\bb{T}, E_\bb{T})$, such that $E_\bb{T}$ only contains \emph{balanced} equations.
Then sequential lifting defines a distributive law $T\grD \rightarrow \grD T$.
\end{thm}

\begin{proof}
We use the fact that a distributive law of the given type is equivalent to a lifting of the monad $\grD$ to the Eilenberg-Moore category of $T$, see Theorem \ref{th:beck:dist-lift}.

By Proposition~\ref{prop:seq-eq-preservation}, sequential lifting defines a lifting of $T$-algebra structures on $X$ to $T$-algebra structures on $\grD X$. It therefore suffices to show that $\eta^{\grD}$ and $\mu^{\grD}$ are homomorphisms of $T$-algebras. We just show this for $\mu^{\grD}$, and only in the case where one of the input steps, say 
$x_i = \step^\kappa(x'_i)$. Then, using Lemma~\ref{lem:takestepsoutanyorder-binary}:
\begin{align*}
 & \mu (\seqinterpretation{\nop}{\grD \grD X} (x_1, \ldots, \step^\kappa(x'_i), \ldots, x_n)) \\
 & = \mu (\step^\kappa(\tabs{\tickA}{\kappa} \seqinterpretation{\nop}{\grD \grD X}(x_1, \ldots, (\tapp[\tickA]{x'_i}), \ldots, x_n))) \\
 & = \step^\kappa(\tabs{\tickA}{\kappa} (\mu (\seqinterpretation{\nop}{\grD \grD X}(x_1, \ldots, (\tapp[\tickA]{x'_i}), \ldots, x_n)))) \\
 & = \step^\kappa(\tabs{\tickA}{\kappa} (\seqinterpretation{\nop}{\grD X}((\mu x_1), \ldots, (\mu (\tapp[\tickA]{x'_i})), \ldots, (\mu x_n)))) \\
% & = \seqinterpretation{\nop}{\grD X}((\mu x_1), \ldots, \step(\tabs{\tickA}{\kappa} (\mu (\tapp[\tickA]{x'_i})), \ldots, (\mu x_n))) \\
 & = \seqinterpretation{\nop}{\grD X}((\mu x_1), \ldots, (\mu x_i), \ldots, (\mu x_n)).
\end{align*}
Using guarded recursion in the third equality.
\end{proof}

Combining this with Lemma~\ref{lem:dist:grD:to:ciD} we obtain a distributive law $T\ciD \to \ciD T$ for all $T$ as in Theorem~\ref{thm:seq-lifting-distlaw}.

\begin{rem}
 Since $\grD$ is a commutative monad, we already know from Manes and Mulry \cite{ManesMulry2007} and Parlant \cite{parlant-thesis} that there is a distributive law in the case where $\bb{T}$ only has \emph{linear} equations. We can extend this linearity requirement here to allow duplications of variables, as long as there are equally many duplicates on either side of each equation.
\end{rem}

\begin{exa}
The sequential distributive law successfully combines the delay monad with the binary tree monad, the list monad, and the multiset monad, resulting in the monads $\grD B$, $\grD L$, and $\grD M$, respectively.
\end{exa}

\section{Idempotent Equations}
\label{sec:nogo-idemp}

This section studies distributive laws $T\ciD \rightarrow \ciD T$ for $T$ an algebraic monad with an idempotent binary operation ``$\binop$''. Since idempotency is not a balanced equation, as remarked after Proposition~\ref{prop:seq-eq-preservation}, sequential distribution does not respect it, and so neither parallel nor sequential distribution define distributive laws in this case. Idempotency turns out to be a tricky equation: We first show two examples of such a theory $T$ where no distributive law $T\ciD \rightarrow \ciD T$ is possible, then a theory where it is, and finally we show that no
distributive law of type $T\grD \rightarrow \grD T$ is possible. First observe the following.

\begin{lem} \label{lem:binopOne}
 Let $T$ be an algebraic monad with an idempotent binary operation $\binop$ and let $\distlaw : T\ciD \rightarrow \ciD T$ be a distributive law.
 There exist binary $T$-operations $\binopOne$ and $\binopPrime$ such that for any $T$-model $X$, the lifting of $\binop$ to $\ciD X$ satisfies
\begin{equation}\label{eq:bopOne-proof}
  \binop(\step x, \step y)  = \step (\binopOne(x,y)),
\end{equation}
and either
\begin{equation}\label{eq:bopPrime1}
  \binop(\step x, y) = \step (\binopPrime(x,y)) \quad \text{and} \quad \binopPrime(x, \step y) = \binopOne(x,y) 
\end{equation}
or
\begin{equation}\label{eq:bopPrime2}
  \binop(\step x, y) = \binopPrime(x,y) \quad \text{and} \quad \binopPrime(x, \step y) = \step (\binopOne(x,y)). 
\end{equation}
Furthermore, $\binopOne$ is idempotent, and associative and/or commutative whenever $\binop$ is so.
\end{lem}

\begin{proof}
To see that there exists a $T$-operation $\binopOne$ satisfying (\ref{eq:bopOne-proof}), consider the following two naturality diagrams, one for the unique map $! : 2 \to 1$ (left) and one for any $f : 2 \rightarrow X$ (right):
\[
 {\small
 \begin{tikzcd}
  \ciD T 2\times \ciD T 2 \ar{r}{\binop} \ar[swap]{d}{\ciD T ! \times \ciD T !} & 
  \ciD T 2 \ar{d}{\ciD T !} & &
  \ciD T 2\times \ciD T 2 \ar{r}{\binop} \ar[swap]{d}{\ciD T f \times \ciD T f} & 
  \ciD T 2 \ar{d}{\ciD T f}  \\
  \ciD T  1\times \ciD T 1 \ar{r}{\binop} & \ciD T 1 & &  
  \ciD T X\times \ciD T X \ar{r}{\binop} & \ciD T X 
 \end{tikzcd}
 }
\]

 Let $\eta^T$ be the unit of $T$, and write $\tt, \ff$ for the two elements of $2$, and $\star$ for the unique element of $1$. Then chasing the element $(\step(\now(\eta^T(\tt))), \step(\now(\eta^T(\ff))))$ along the down, then right path of the first naturality diagram yields (using idempotence of $\binop$):
\begin{align*}
 & \binop(\step(\now(\eta^T(\star))), \step(\now(\eta^T(\star)))) = \step(\now(\eta^T(\star))).
\end{align*}
So then by following the right, then down path of this naturality diagram:
\begin{align} \label{eq:bopOne:element}
\binop(\step(\now(\eta^T(\tt))), \step(\now(\eta^T(\ff)))) = & \step(\now(\binopOne(\tt, \ff))),
\end{align}
for some element $\binopOne(\tt,\ff)$ in $T(2)$ satisfying $T(!)(\binopOne(\tt,\ff)) = \eta^T(\star)$.
Such an element corresponds to an operation $T \times T \to T$, for which we shall also write $\binopOne$.
We then use the second naturality diagram to extend (\ref{eq:bopOne:element}) to any algebra\footnote{
Note that $\binopOne$ on the right hand side of (\ref{eq:bopOne:general}) should be interpreted in $X$ using the algebra structure.} $X$ for $T$, by choosing, for any $x,y : X$, $f(\tt) = x$ and $f(\ff) = y$:
\begin{equation} \label{eq:bopOne:general}
\binop(\step(\now(x)), \step(\now(y))) = \step(\now(\binopOne(x,y))).
\end{equation}

To conclude Equation~(\ref{eq:bopOne-proof}), we use that the multiplication $\mu$ for $\ciD$ is a homomorphism. That is, for any $x,y: \ciD X$:
\begin{align*}
 \mu(\binop(\step(\now(x)), \step(\now(y))))
& =  \binop(\mu(\step(\now(x))), \mu(\step(\now(y)))).
\end{align*}
The left hand side of this equation rewrites to:
\begin{align*}
 \mu(\binop(\step(\now(x)), \step(\now(y))))
& =  \mu(\step(\now(\binopOne(x,y)))) \\
& =  \step(\binopOne(x,y)),
\end{align*}
and the right hand side to:
\begin{align*}
  \binop(\mu(\step(\now(x))), \mu(\step(\now(y))))
  & =  \binop(\step(x), \step(y)),
\end{align*}
which, indeed, proves Equation~(\ref{eq:bopOne-proof}).

Idempotence, associativity, and commutativity of $\binopOne$ now follow from (\ref{eq:bopOne-proof}). We show idempotence. Below, we first use that $\binop$ is idempotent, then apply (\ref{eq:bopOne-proof}), and finally use that $\now$ is a $T$-algebra homomorphism.
\begin{align*}
  \step (\now x) & = \binop(\step (\now x), \step (\now x)) \\
   & = \step (\binopOne (\now x, \now x)) \\
   & = \step (\now (\binopOne(x, x)))
\end{align*}
By injectivity of $\step$ and $\now$, we can therefore conclude that $\binopOne(x, x) = x$ in all $T$-algebras. In other words, $\binopOne$ is idempotent. Associativity and commutativity are shown similarly.

Finally, we want to show that there exists a $T$-operation $\binopPrime$ such that either Equations~(\ref{eq:bopPrime1}) or (\ref{eq:bopPrime2}) hold.
To see this, define for any $T$-algebra $X$ an operation $\binopPrimePrime: \ciD X \times \ciD X \to \ciD X$ by:
 \[
  \binopPrimePrime(x,y) = \binop(\step(x),y).
 \]
 We then have: $\binopPrimePrime(x, \step(y)) = \binop(\step(x), \step(y)) = \step(\binopOne(x,y))$.
 So it follows that: either $\binopPrimePrime(x,y) = \step(\binopPrime(x,y))$ for some operation $\binopPrime$ on $T$ such that:
 \[
 \binopPrime(x,\step(y)) = \binopOne(x,y),
 \]
 or $\binopPrimePrime(x,y) = \binopPrime(x,y)$ for some operation $\binopPrime$ on $T$ such that:
 \[
 \binopPrime(x,\step(y)) = \step(\binopOne(x,y)). \qedhere
 \]
\end{proof}

Lemma~\ref{lem:binopOne} provides the key step in ruling out distributive laws for both the finite powerset monad over the delay monad and the finite distributions monad over the delay monad.

\subsection*{No Distributive Law for Powerset}
In Example~\ref{ex:boom-monads-alg}, we introduced finite powerset as the free monad of the algebraic theory of semilattices. While this monad exists as a consequence of
the general constructions of Appendix~\ref{app:encoding:alg:th}, it is convenient to work with a more direct construction as a HIT generated by singletons, $\cup$, $\emptyset$,
and paths for the equations as well as set-truncation~\cite{Geuvers2018}. Using this, one gets the usual induction principle for $\Pfin$ as a consequence of HIT-induction:
If $\phi : \Pfin(X) \to \Prop$, to prove that $\phi(A)$ holds for all $A$, it suffices to show $\phi(\emptyset)$, $\phi(\{x\})$ for all $x$, and that $\phi(A)$ and $\phi(B)$
implies $\phi(A\cup B)$.

The proof of Proposition~\ref{prop:nodistpowerset} below uses contradictions of various kinds. As we are in constructive type theory, we would like to briefly explain how we arrive at these contradictions. These methods are also relevant for the proof of Proposition~\ref{prop:nodistprobability}. 

\begin{rem}\label{rem:contr}
 We derive contradictions in two ways.
\begin{itemize}
      \item We conclude that a certain equation must be true in all $\Pfin$- (or $\Dfin$)-algebras. In this case, we have an obvious algebra in mind where the equation is provably false. For instance, the \\[0.5ex]
          \begin{minipage}{0.80\textwidth}
           equation $x \cup y = x$ for the union operation of $\Pfin$. Take the $\Pfin$-algebra consisting of 4 distinct elements, arranged in a semilattice. We interpret $\emptyset$ as the bottom element, $0$, and $\cup$ as the least upper bound. The four element set can be constructed as a $4$-fold sum of the unit type by itself, and so has decidable equality. 
      The equation $1 \cup 2 = 3$ therefore implies false.\\
          \end{minipage}
          \begin{minipage}{0.15\textwidth}
          \begin{center}
          \scalebox{0.55}{
            \begin{tikzcd}[ampersand replacement=\&]
             \& 3 \& \\
             1 \ar[ru, dash] \ar[rd, dash] \&  \& 2 \ar[lu, dash] \ar[ld, dash] \\
             \& 0 \& \\
            \end{tikzcd}
          }
          \end{center}
          \end{minipage}
        
  \item We conclude that $\step x = \now y$ for some $x$ and $y$. This gives a contradiction because $\step$ and $\now$ are the $\inr$ and $\inl$ injections into a coproduct, which are never equal \cite[Eq.2.12.3]{hottbook}. 
\end{itemize}
\end{rem}

\begin{prop}\label{prop:nodistpowerset}
 There is no distributive law $\Pfin \ciD \to \ciD\Pfin$ for $\Pfin$ the finite powerset functor.
\end{prop}

\begin{proof}
 Assuming that a distributive law exists, we can apply Lemma~\ref{lem:binopOne} to the union operation $\cup$ of $\Pfin$. We see there are only 
 four possible cases for the obtained $\binopOne(x,y)$ and $\binopPrime(x,y)$: $\emptyset, x, y$ and $x \cup y$ (this can be proved by HIT-induction). For $\binopOne$, the first three possibilities are immediately ruled out by the fact that $\binopOne$ is idempotent and commutative (since $\cup$ is). Therefore we must have $\binopOne(x,y) = x \cup y$.

We show that also $\binopPrime(x,y) = x \cup y$.

First, notice that:
\begin{align}\label{eq:binopPrime-now-step}
  \binopPrime(\now x, \step( \now y)) = & \begin{cases}
        \now (x \cup y), & \mbox{in case of Equation~(\ref{eq:bopPrime1})}, \\
        \step(\now (x \cup y)), & \mbox{in case of Equation~(\ref{eq:bopPrime2})}.
      \end{cases}
\end{align}
This follows from, in case of (\ref{eq:bopPrime1}):
\begin{align*}
\binopPrime(\now x, \step( \now y)) & = \binopOne(\now x, \now y) \\
                                    & = \now \binopOne(x, y)  \\
                                    & = \now (x \cup y).
\end{align*}
and in case of (\ref{eq:bopPrime2}):
\begin{align*}
\binopPrime(\now x, \step( \now y)) & = \step (\binopOne(\now x, \now y)) \\
                                    & = \step (\now \binopOne(x, y))  \\
                                    & = \step (\now (x \cup y)).
\end{align*}

Equation~(\ref{eq:binopPrime-now-step}) allows us to rule out $\binopPrime(x,y) = \emptyset$, $\binopPrime(x,y) = x$, and $\binopPrime(x,y) = y$.
 
If $\binopPrime(x,y) = \emptyset$, then $\binopPrime(\now x, \step( \now y)) = \emptyset = \now \emptyset$, because $\now$ is a $\Pfin$-algebra homomorphism.
So from (\ref{eq:binopPrime-now-step}), we conclude either $\now(x \cup y) = \now \emptyset$, 
which in turn implies that $x \cup y = \emptyset$ in all $\Pfin$-algebras, 
or $\step (\now(x \cup y)) = \now \emptyset$. Both of these are contradictions, as explained in Remark~\ref{rem:contr}.

If $\binopPrime(x,y) = x$, then $\binopPrime(\now x, \step( \now y)) = \now x$, and from (\ref{eq:binopPrime-now-step}) we conclude $\now x = \now (x \cup y)$ and hence $x \cup y = x$ in all $\Pfin$-algebras in case of (\ref{eq:bopPrime1}), or $\now x = \step (\now (x \cup y))$ in case of (\ref{eq:bopPrime2}), both of which are again contradictions.

If $\binopPrime(x,y) = y$, then $\binopPrime(\now x, \step( \now y)) = \step (\now y)$, and from (\ref{eq:binopPrime-now-step}) we conclude $\step (\now y) = \now(x \cup y)$ in case of (\ref{eq:bopPrime1}), or $\step (\now y) = \step (\now (x \cup y))$ in case of (\ref{eq:bopPrime2}), which then implies that $y = x \cup y$ in all $\Pfin$-algebras. So again, both cases lead to a contradiction.

So we conclude that also $\binopPrime(x,y) = x \cup y$. From this we now prove                   
\[
\step(x) \cup y = \step(x \cup y).
\]
In the case of (\ref{eq:bopPrime1}), this is the first equation of (\ref{eq:bopPrime1}), where we use that $\binopPrime(x, y) = x \cup y$. In the case of (\ref{eq:bopPrime2}), it is proved as follows
\[
\step(x) \cup y = y \cup \step(x) = \binopPrime(y, \step(x)) = \step(\binopOne(y,x)) = \step(y \cup x) = \step(x \cup y).
\]

Finally the contradiction proving the proposition can be obtained as follows:
\[
  \step(x) = \step(x) \cup \step (x) = \step(x \cup \step(x)) = \step(\step x \cup x) = \step^2(x).  \qedhere
\]
\end{proof}

\subsection*{No Distributive Law for Probability Distributions}
We now turn to the probability distributions monad. Again, we formalise this monad as a HIT, $\Dfin$, using the constructors and equalities of the theory of convex algebras~\cite{Stassen2024}. For the unit interval, we restrict to the rational unit interval, which we also denote by $\II$. 
This provides us with decidable equality on $\II$, which is essential to our proof.

\begin{prop}\label{prop:nodistprobability}
 There is no distributive law $\Dfin \ciD \to \ciD\Dfin$ for $\Dfin$ the finite distributions functor.
\end{prop}

\begin{proof}[Proof]
 We apply Lemma~\ref{lem:binopOne} to $\DistChoice{\frac{1}{2}}{\step(x)}{\step(y)}$. There are now only three possible cases for the obtained $\binopOne(x,y)$ and $\binopPrime(x,y)$: $x, y,$ and $\DistChoice{p}{x}{y}$ for some $p \in \II$ (this can be proved by HIT-induction). Just like in the case for powerset, we know that $\binopOne$ is commutative because $\DistChoice{\frac{1}{2}}{}{}$ is so, which immediately rules out $\binopOne(x,y) = x$ and $\binopOne(x,y) = y$. Hence $\binopOne(x, y) = \DistChoice{p}{x}{y}$ for some $p \in \II$. In fact, we must have $\binopOne(x, y) = \DistChoice{\frac{1}{2}}{x}{y}$. To see this, consider the closed unit interval $\IIc$ as a $\Dfin$ model, where the terms $\DistChoice{q}{x}{y}$ are interpreted as $q * x + (1-q) * y$. Because of commutativity of $\binopOne$, we then have:
 \begin{align*}
 p & = \DistChoice{p}{1}{0} = \binopOne(1, 0) = \binopOne(0, 1) = \DistChoice{p}{0}{1} = 1 - p
\end{align*}    
Since $\II$ has decidable equality, we know that $p = 1 - p$ is only true for $p = \frac{1}{2}$. Therefore, we must have $\binopOne(x,y) = \DistChoice{\frac{1}{2}}{x}{y}$.

For $\binopPrime(x,y)$, the possibilities $\binopPrime(x,y) = x$ and $\binopPrime(x,y) = y$ are also ruled out in the same way as they were in the proof for powerset. We conclude that $\binopPrime(x,y) = \DistChoice{p}{x}{y}$ for some $p \in (0,1)$, but now we do not necessarily have $p = \frac{1}{2}$.

Putting this information in Equations~(\ref{eq:bopPrime1}) and (\ref{eq:bopPrime1}) tells us that either
\begin{equation}\label{eq:bopPrime1-DistChoice}
  \DistChoice{\frac{1}{2}}{\step x}{y} = \step (\DistChoice{p}{x}{y}) \quad \text{and} \quad \DistChoice{p}{x}{\step y} = \DistChoice{\frac{1}{2}}{x}{y} 
\end{equation}
or
\begin{equation}\label{eq:bopPrime2-DistChoice}
  \DistChoice{\frac{1}{2}}{\step x}{y} = \DistChoice{p}{x}{y} \quad \text{and} \quad \DistChoice{p}{x}{\step y} = \step (\DistChoice{\frac{1}{2}}{x}{y}). 
\end{equation}

We show that the equations of (\ref{eq:bopPrime1-DistChoice}) contradict each other. Consider the term 
\[\DistChoice{1-p}{\step (\now x)}{\now y}.\]
We see that the second equation of (\ref{eq:bopPrime1-DistChoice}), via symmetry, implies that 
\[\DistChoice{1-p}{\step (\now x)}{\now y} = \DistChoice{p}{\now y}{\step(\now x)}= \DistChoice{\frac{1}{2}}{\now y}{\now x} = \now ( \DistChoice{\frac{1}{2}}{x}{y}).\]
We will now use the first equation to show that $\DistChoice{1-p}{\step (\now x)}{\now y}$ in fact takes a step.

First, we use that there exists an $N$ such that $(\frac{1}{2})^N \leq (1-p)$. Then we can write:
\begin{align*}
 & \DistChoice{1-p}{\step (\now x)}{\now y} = \DistChoice{(\frac{1}{2})^N}{\step(\now x)}{z}, \\
& \text{where } z = \DistChoice{\frac{(1-p) - (\frac{1}{2})^N}{1 - (\frac{1}{2})^N}}{(\step (\now x))}{(\now y)}. 
\end{align*}
We show by induction on $n$ that for all $n$ and all $x', z'$:
\[
\DistChoice{(\frac{1}{2})^n}{\step x'}{z'} = \step (\DistChoice{p^n}{x'}{z'}).
\]
For $n = 1$, this is the first equation of (\ref{eq:bopPrime1-DistChoice}). For the induction step, we rewrite $\DistChoice{(\frac{1}{2})^{n+1}}{\step x'}{z'}$, apply the first equation of (\ref{eq:bopPrime1-DistChoice}), and finally the induction hypothesis. 
\begin{align*}
\DistChoice{(\frac{1}{2})^{n+1}}{\step x'}{z'} & = \DistChoice{(\frac{1}{2})^n}{(\DistChoice{\frac{1}{2}}{\step x'}{z'})}{z'} \\
& = \DistChoice{(\frac{1}{2})^n}{\step(\DistChoice{p}{x'}{z'})}{z'} \\
& = \step(\DistChoice{p^n}{(\DistChoice{p}{x'}{z'})}{z'}) \\
& = \step(\DistChoice{p^{n+1}}{x'}{z'})
\end{align*}

Therefore, we have $\DistChoice{1-p}{\step (\now x)}{\now y} = \DistChoice{(\frac{1}{2})^N}{\step(\now x)}{z} = \step(\DistChoice{p^N}{\now x}{z})$, contradicting our earlier conclusion that $\DistChoice{1-p}{\step (\now x)}{\now y} = \now (\DistChoice{\frac{1}{2}}{x}{y})$. (See Remark~\ref{rem:contr}).
Proving that the two equations of (\ref{eq:bopPrime2-DistChoice}) contradict each other works similarly.
\end{proof}

\subsection*{Distributive Laws in General: Yes and No.}
Lemma~\ref{lem:binopOne} does not rule out distributive laws for \emph{all} monads with idempotent binary operations.
\begin{exa} \label{ex:dist:law:idempotent}
Let $A$ be the algebraic theory with one idempotent binary operation $*$ and one unary operator $!$, with no further equations. Let $T$
be the monad generated by $A$. There is a distributive law $\distlaw : T\ciD \to \ciD T$ given by the following clauses
\begin{align*}
 ! (\step (x)) & = x & \step(x) * y & = \step(x * (! y)) & x * \step(y) & = \step((!x) * y).
\end{align*}
Note in particular, that $\step(x) * \step(y) = \step(x * (!(\step(y)))) = \step( x * y)$. Even though this seems to implement parallel computation, the unary $!$ ensures that we cannot `hide' steps by putting them in different layers of $DDX$, like we did in the proof of Theorem~\ref{thm:mogelvezzo}: remember that applying $\mu^D \circ D \distlaw \circ \distlaw$ to $\step \now (\now x) * \now (\step \now y)$ resulted in two steps: $\step \step \now (x * y)$, because the steps were not executed in parallel. But here, $!$ acts as a messenger between the two layers of $D$, cancelling the second step:
\begin{align*}
\step \now (\now x) * \now (\step \now y) & = \step( (\now (\now x)) * (! \now (\step \now y))) \\
& = \step ((\now (\now x)) * (\now ! (\step \now y))) \\
& = \step \now ((\now x) * ! (\step \now y)) \\
& = \step \now ((\now x) * (\now y)) \\
& = \step \now (\now (x * y)). 
\end{align*}
This example can be extended to $*$ associative, if
the equation $!(x * y) = (!x) * (!y)$ is added.
\end{exa}

Example~\ref{ex:dist:law:idempotent} cannot be adapted to give a distributive law $T\grD \rightarrow \grD T$, because in the corresponding version of the first rule $! (\step^\kappa (x)) = x$, the $x$ on the left hand side would have type $\later^\kappa (\grD X)$, but needs type $\grD X$ on the right hand side. In fact, no such distributive law is possible.

\begin{thm}[No-Go Theorem]
 \label{thm:gen:no:go}
Let $T$ be a monad with a binary algebraic operation that is commutative and idempotent.
Then there is no distributive law of type $T\grD \rightarrow \grD T$.
\end{thm}

\begin{proof}
Given the algebraic monad $T$, presented by a theory with a binary operation $\binop$ that is commutative and idempotent, we assume that there is a distributive law of type $T\grD \rightarrow \grD T$, and derive a contradiction.

By Lemma \ref{lem:dist:grD:to:ciD}, such a distributive law induces a distributive law of type $T \ciD \rightarrow \ciD T$.
Then by Lemma \ref{lem:binopOne} we know that there exist binary operations $\binopOne$ and $\binopPrime$ such that for any $T$-model $X$, the lifting of $\binop$ to $\ciD X$ satisfies:
\begin{equation}\label{eq:stepstepbinopOne}
\binop(\step x, \step y)  = \step (\binopOne(x,y)),
\end{equation}
and either
\begin{align}\label{eq:nogoGR1}
   \binop(\step x, y)  & = \step (\binopPrime(x,y)) &
   \binopPrime(x, \step y)  & = \binopOne(x,y),
\end{align}
or
\begin{align}\label{eq:nogoGR2}
   \binop(\step x, y)  & = \binopPrime(x,y) &
   \binopPrime(x, \step y)  & = \step (\binopOne(x,y)).
\end{align}
We show that both cases lead to a contradiction.

First we assume that we are in case (\ref{eq:nogoGR1}).
The fact that the distributive law $\distlaw : T \ciD \rightarrow \ciD T$ is induced by a distributive law
$\distlaw^\kappa: T \grD \rightarrow \grD T$ means that for any $X$, and any $x : T(\ciD X)$
\[
 \capp{\distlaw (x)} = \distlaw^\kappa(T( \mathsf{ev}_\kappa) x)
\]
where $T( \mathsf{ev}_\kappa)$ is the functorial action of $T$ applied to the map $\mathsf{ev}_\kappa : \ciD X \to \grD X$,
mapping $y$ to $\capp y$. From this and (\ref{eq:nogoGR1}) it follows that for any $T$-model $X$, the
$T$-algebra structure on $\grD X$ induced by $\distlaw^\kappa$ satisfies
\begin{align} \label{eq:binop:causal}
   \binop(\step^\kappa (\tabs\alpha{\kappa}{(\capp x)}), \capp y)  & = \step^\kappa(\tabs{\alpha}{\kappa}(\binopPrime(\capp x,\capp y))).
\end{align}
for any $x,y : \ciD X$.

Now consider the case of $X = T(2)$. We write $\tt, \ff$ for the two elements of $2$ and, leaving the unit of $T$ implicit, also for
the corresponding two elements of $T(2)$. To arrive at a contradiction, we will prove $\later^\kappa \bot$ by guarded recursion.
Since the proof is parametric in $\kappa$, we can use it to prove $\forall\kappa . \later^\kappa \bot$, which implies
$\forall\kappa . \bot$, in turn implying $\bot$ since $\bot$ is clock irrelevant.

Note that since $\tt = \ff$ is equivalent to $\bot$,
by extensionality for $\later^\kappa$ and injectivity of $\step^\kappa$ and $\now^\kappa$, $\later^\kappa \bot$ is equivalent to
\begin{equation} \label{eq:step:tt:eq:ff}
\step^\kappa(\tabs\alpha{\kappa}{\now^\kappa(\tt)}) = \step^\kappa(\tabs\alpha{\kappa}{\now^\kappa(\ff)}),
\end{equation}
where the equality is between elements of $\grD T(2)$.
So we will prove (\ref{eq:step:tt:eq:ff}) by guarded recursion. First note
\begin{align*}
  \step^\kappa(\tabs\alpha{\kappa}{\now^\kappa(\tt)})
   & = \binop (\step^\kappa(\tabs\alpha{\kappa}{\now^\kappa(\tt)}), \step^\kappa(\tabs\alpha{\kappa}{\now^\kappa(\tt)})) \\
                 & = \step^\kappa (\tabs{\alpha}{\kappa}(\binopPrime(\now(\tt), \step^\kappa(\tabs\alpha{\kappa}{\now^\kappa(\tt)}))) \\
                 & = \step^\kappa (\tabs{\alpha}{\kappa}(\binopPrime(\now(\tt), \step^\kappa(\tabs\alpha{\kappa}{\now^\kappa(\ff)})))) \\
   & = \binop (\step^\kappa(\tabs\alpha{\kappa}{\now^\kappa(\tt)}), \step^\kappa(\tabs\alpha{\kappa}{\now^\kappa(\ff)})),
\end{align*}
using idempotency of $\binop$, (\ref{eq:binop:causal}), and the guarded recursion assumption. By a similar argument
\[
 \step^\kappa(\tabs\alpha{\kappa}{\now^\kappa(\ff)}) =
 \binop (\step^\kappa(\tabs\alpha{\kappa}{\now^\kappa(\ff)}), \step^\kappa(\tabs\alpha{\kappa}{\now^\kappa(\tt)})),
\]
and so by commutativity of $\binop$, (\ref{eq:step:tt:eq:ff}) follows as desired.

Suppose now (\ref{eq:nogoGR2}). Again we want to arrive at a contradiction by considering the $T$-algebra structures
on $\ciD X$ and $\grD X$ induced by one on $X$. %First 
Note that $\binopOne$ is commutative and idempotent by Lemma \ref{lem:binopOne}.

We show that
$\binopOne(x,y) = \binopOne(z,y)$ for distinct variables $x,y,z$ leads to a contradiction. This can be read as an equation
between elements in $T(3)$. Since we can substitute for $y$ in the assumed
equation we get
\begin{align*}
 x = \binopOne(x,x) = \binopOne(z,x) = \binopOne(x,z) = \binopOne(z,z) = z,
\end{align*}
which is a contradiction, since $x$ and $z$ are distinct. Note that this argument uses that $3$ has decidable equality.

Similarly, also the equation $\binopPrime(x, y) = \binopPrime(z, y)$ for distinct variables $x,y,z$ leads to a contradiction, since
by (\ref{eq:nogoGR2}) this leads to
\begin{align*}
 \step(\now \binopOne(x,y)) & = \step(\binopOne(\now(x), \now(y))) \\
 & = \binopPrime(\now(x), \step (\now(y))) \\
 & = \binopPrime(\now(z), \step (\now(y))) \\
 & =  \step(\now \binopOne(z,y)),
\end{align*}
and so $\binopOne(x,y) = \binopOne(z,y)$, which we have just proven inconsistent.

Similarly to the argument in the case of  (\ref{eq:nogoGR1}) we can show that the fact that the distributive law is induced by one on $T\grD \to \grD T$ implies
that the model structure on $\grD X$ for any $T$-model $X$ satisfies
\begin{equation} \label{eq:op:op':guarded}
\binop(\step^\kappa (\tabs{\alpha}{\kappa}{\capp x}), \capp y)  = \binopPrime(\capp x,\capp y),
\end{equation}
for any $x,y: \ciD X$.

We use the above to prove $\bot$ by guarded recursion. So assume $\later^\kappa\bot$. As we have seen, to prove $\bot$, it suffices to show
$\binopPrime(x, y) = \binopPrime(z, y)$ for distinct variables $x,y,z$. We prove that $\now^\kappa(\binopPrime(x, y)) = \now^\kappa(\binopPrime(z, y))$ as elements
of $\grD(T(3))$:
\begin{align*}
 \now^\kappa(\binopPrime(x, y))  & = \binopPrime(\now^\kappa(x), \now^\kappa(y)) \\
 & = \binop(\step^\kappa(\tabs{\alpha}{\kappa}{\now^\kappa(x)}), \now^\kappa(y)) \\
 & = \binop(\step^\kappa(\tabs{\alpha}{\kappa}{\now^\kappa(z)}), \now^\kappa(y)) \\
 & = \binopPrime(\now^\kappa(z), \now^\kappa(y)) \\
 & = \now^\kappa(\binopPrime(z, y)),
\end{align*}
where we have applied (\ref{eq:op:op':guarded}) to $\Lambda\kappa . \now^\kappa(x)$ and $\Lambda\kappa . \now^\kappa(y)$
in the second equality and used the assumption $\later^\kappa\bot$ in the third. 
\end{proof}

\begin{cor}\label{cor:no:go:continuation}
There is no distributive law of type $\mathcal{C}\grD \rightarrow \grD \mathcal{C}$, distributing the continuations monad $\mathcal{C}X = (X \rightarrow R) \rightarrow R$ over the guarded delay monad, whenever $R$ has a binary operation that is commutative and idempotent.
\end{cor}

\begin{proof}
 If $R$ has a binary operation that is commutative and idempotent, so does $\mathcal C$~\cite{hyland2007}. The result now
 follows directly from Theorem~\ref{thm:gen:no:go}. 
\end{proof}

\section{Semi-Go Theorem: Up to Weak Bisimilarity}\label{sec:weak:bisim}

In the proof of Theorem~\ref{thm:gen:no:go} the failure of existence of distributive laws comes down to a miscounting of steps. This section shows that this is indeed
all that fails, and that parallel lifting defines a distributive law \emph{up to weak bisimilarity} for algebraic monads with no drop equations.
Weak bisimilarity is a relation on the coinductive delay monad, which relates computations that only differ by a finite number of steps. To make this precise, we have to work in a category of setoids, as simply quotienting $\ciD$ up to weak bisimilarity does not yield a monad [CUV19]. The objects are pairs
$(X,R)$, where $R$ is an equivalence relation on $X$, and morphisms are equivalence classes of maps $f$ between the underlying sets respecting the relations. Two such maps are equivalent if their values on equal input are related by the equivalence relation on the target type.

We first define a lifting of the coinductive delay monad $\ciD$ to the category of setoids. We do this via a similar relation
(taken from M{\o}gelberg and Paviotti~\cite{Paviotti:FPC:journal})
defined for the guarded delay monad, because that allows us to reason using guarded recursion. 
We write $\nextstep$ for $\step^\kappa \circ \nextop^\kappa : \grD X \to \grD X$. 

\begin{defiC}[\cite{Paviotti:FPC:journal}]\label{def:wbisim:gr}
  Let $X,Y$ be sets, and suppose $R : X \to Y \to \Prop$ is a relation. Define weak bisimilarity up to $R$, written
$\wbsimgr[R] : \grD X \to \grD Y \to \Prop$, by:
\begin{align*}
 \now^\kappa (x) \wbsimgr[R] y~\defeq \; & \exists (n : \N, y': Y) .
                                   y = (\nextstep)^n (\now^\kappa (y')) \text{ and } R(x,y'), \\
  x \wbsimgr[R] \now^\kappa(y)~\defeq \; & \exists (n : \N, x': X) .
                                  x = (\nextstep)^n (\now^\kappa(x')) \text{ and } R(x',y), \\
 \step^\kappa (x) \wbsimgr[R] \step^\kappa(y)~\defeq \; & \latbind\tickA\kappa (\tapp x \wbsimgr[R] \tapp y).
\end{align*}
\end{defiC}
Note that the two first cases both apply for $\now^\kappa(x) \wbsimgr[R] \now^\kappa (y)$, but that they are equivalent in that case.

\begin{defi} \label{def:wbisim:ci}
Let $R : X \to Y \to \Prop$ be a relation.
Define $ \wbsim[R]  : \ciD X \to \ciD Y \to \Prop$ as
\begin{align*}
 x \wbsim[R] y \defeq \forall\kappa . \capp x \wbsimgr[R] \capp y.
\end{align*}
\end{defi}

The above definition is an encoding (using guarded recursion) of the standard coinductive definition of weak bisimilarity: 
\begin{align*}
 \now (x) \wbsim[R] y~\defeq \; & \exists (n : \N, y': Y) .
                                   y = (\step)^n (\now (y')) \text{ and } R(x,y'), \\
  x \wbsim[R] \now(y)~\defeq \; & \exists (n : \N, x': X) .
                                  x = (\step)^n (\now(x')) \text{ and } R(x',y), \\
 \step (x) \wbsim[R] \step(y)~\defeq \; & x \wbsim[R] y.
\end{align*}
As one would expect, weak bisimilarity is closed under adding steps on either side:

\begin{lemC}[\cite{Paviotti:FPC:journal}] \label{lem:wbsimgr:step}
 If $x\wbsimgr[R] y$ then $\nextstep(x)\wbsimgr[R] y$.  As a consequence $x \wbsim[R] y$ implies $\step(x) \wbsim[R] y$. 
 If $\step^\kappa(x)\wbsimgr[R] y$ then ${\latbind\tickA\kappa (\tapp x\wbsimgr[R] y)}$.
\end{lemC}

\begin{proof}
 The proof of the first statement is by cases of $x$ and $y$. If $y = \now^\kappa(y')$, then the proof is easy. If $y = \step^\kappa(y')$ and $x = \now^\kappa(x')$, then there exists
 $n$ and $y''$ such that $(\nextstep)^n(\now^\kappa(y'')) = y$, and $R(x, y'')$. We must show that
 $\latbind\tickA\kappa (\now^\kappa(x) \wbsimgr[R] \tapp {y'})$, but since $\tapp{y'} = (\nextstep)^{n-1}(\now^\kappa (y''))$, this is easy. Finally, if
 $y = \step^\kappa(y')$ and $x = \step^\kappa(x')$ we must show that $\latbind\tickA\kappa (\step^\kappa(x') \wbsimgr[R] \tapp {y'})$ assuming
 $\latbind\tickA\kappa (\tapp{x'} \wbsimgr[R] \tapp {y'})$, which follows from guarded recursion using that tick-irrelevance implies
\begin{align*}
  \nextstep(\tapp{x'}) &= \step^\kappa(\tabs\tickB\kappa (\tapp{x'})) \\
  & = \step^\kappa(\tabs\tickB\kappa (\tapp[\tickB]{x'})) \\
  & = \step^\kappa(x')
\end{align*}

For the last statement, suppose $\step^\kappa(x)\wbsimgr[R] y$. Then by the first statement $\step^\kappa(x)\wbsimgr[R] \nextstep(y)$ which unfolds to
 $\latbind\tickA\kappa (\tapp x\wbsimgr[R] y)$.
\end{proof}

If $R$ is symmetric and reflexive, then the same properties hold for $\wbsimgr[R]$, but transitivity is not preserved.
In fact, if $\wbsimgr[\eq]$ were transitive, then one could prove that $\divergence^\kappa \wbsimgr[\eq] \now^\kappa(x)$ for any $x$ by guarded recursion,
where $\divergence^\kappa = \fix(\step^\kappa)$ is the divergent term:
Suppose $\later(\divergence^\kappa \wbsimgr[\eq] \now^\kappa(x))$. Then $\nextstep (\divergence^\kappa)\wbsimgr[\eq] \nextstep(\now^\kappa(x)))$, and so by 
Lemma~\ref{lem:wbsimgr:step} and transitivity also $\divergence^\kappa \wbsimgr[\eq] \now^\kappa(x)$. But,  
$(\Lambda\kappa. \divergence^\kappa) \wbsim[\eq] \now(x)$ does not hold, so $\wbsimgr[\eq]$ cannot be transitive. 
Danielson~\cite{Danielsson18} made a similar observation in the setting of sized types. 
On the other hand, $\wbsim[\eq]$ \emph{is} transitive:

\begin{lem}\label{lem:wbsim:trans}
 If $R$ is an equivalence relation, so is $\wbsim[R]$.
\end{lem}

\begin{proof}
 We just show transitivity.  Say $x$ terminates to $x'$ if $x = \step^n (\now x')$. 
 It is easy to see 
 that if $x \wbsim[R] y$ and $x$ terminates to $x'$, then also $y$ terminates to a $y'$ related to $x'$ in $R$. 
 On the other hand, if 
 $x$ and $y$ both terminate to related elements then $x \wbsim y$. Now suppose $x \wbsim y$ and $y \wbsim z$.
  
 The proof now proceeds by guarded recursion to show that $\capp x \wbsimgr \capp z$.
 If one of $x,y,z$ is of the form $\now(w)$, then all of them terminate to related elements, implying that 
 $x \wbsim z$. Suppose now that 
\begin{align*}
 x & = \step(x') &
 y & = \step(y') &
 z & = \step(z') 
\end{align*}
 and $x' \wbsim[R] y'$ and $y' \wbsim[R] z'$. By guarded recursion, we then get
\[
 \latbind\tickA\kappa (\capp{x'} \wbsimgr[R] \capp{z'})
\]
which is equivalent to $\capp x \wbsimgr[R] \capp z$ as desired. 
\end{proof}

We note the following, which was also observed by Chapman et al~\cite{quotientingDelay}.

\begin{prop} \label{prop:monad:setoids}
 The mapping $\ciDsetoid(X,R) = (\ciD X, \wbsim[R])$ defines a monad on the category of setoids.
\end{prop}

Proposition~\ref{prop:monad:setoids} follows directly from the following lemma, which we will also need later.

\begin{lem}
 If $x,y : \grD(\grD(X))$ and $x\wbsimgr y$
 (omitting the subscript) then $\mu(x)\wbsimgr \mu(y)$
\end{lem}

\begin{proof}
 This is by guarded recursion and cases of $x$ and $y$.
 If $x = \now(x')$ we get $n$ such that $y = (\nextstep)^n (\now(y'))$ and $x' \wbsimgr[R] y'$. We must show
 $x' \wbsimgr[R] (\nextstep)^n(y')$, which is just an $n$-fold application of Lemma~\ref{lem:wbsimgr:step}. The case of $y = \now(y')$ is similar,
 and the case of both $x$, $y$ of the form $\step$ follows by guarded recursion.
\end{proof}

Similarly, any algebraic monad $T$ can be lifted to the category of setoids by defining $T(R)$ to be the smallest equivalence relation relating an equivalence class
$[t(x_1,\dots, x_n)]$ to $[t(y_1,\dots, y_n)]$ if $R(x_i, y_i)$ for all $i$. We write $\Tsetoid$ for this.

Recall that by Proposition \ref{prop:par-eq-preservation},
if $T$ is an algebraic monad given by a theory with no drop equations, then
parallel lifting defines a natural transformation $T\ciD \to \ciD T$ on the category of sets. We show that this map
lifts to a distributive law on the category of setoids.

\begin{thm} \label{thm:dist:law:setoid}
 Let $T$ be the free model monad of algebraic theory $\bb{T} = (\Sigma_\bb{T}, E_\bb{T})$, such that $E_\bb{T}$ contains no drop equations.
Then parallel lifting defines a distributive law of monads $\Tsetoid\ciDsetoid \to \ciDsetoid\Tsetoid$.
\end{thm}

\begin{rem}
 The free monad on an algebraic theory could alternatively be expressed on the category of setoids
 by taking the set to be the free monad just on operations, introducing the equations of the theory
 into the equivalence relation. In the presence of the axiom of choice this generates a monad equivalent
 to $\Tsetoid$, and we expect that the proof of Theorem~\ref{thm:dist:law:setoid} be adapted to that choice as well.
\end{rem}

Before proving Theorem~\ref{thm:dist:law:setoid}, we first establish that parallel lifting defines a natural transformation.

\begin{lem} \label{lem:seq:lifting:setoid}
 Let $T$ be the free model monad of algebraic theory $\bb{T} = (\Sigma^\bb{T}, E^\bb{T})$, such that $E^\bb{T}$ contains no drop equations.
 Parallel lifting defines a natural transformation $\Tsetoid\ciDsetoid \to \ciDsetoid\Tsetoid$.
\end{lem}

\begin{proof}
 We show that if $x_i \wbsim[R] y_i$ for $i = 1,\dots, n$, then $\parinterpretation{t}{}(x_1, \dots, x_n) \wbsim[T(R)] \parinterpretation{t}{}(y_1, \dots, y_n)$.
 We do this in the special case of $t$ being an operation from the signature, which then generalises by induction to general $t$. We show the special case by proving the stronger statement that if $x_i \wbsimgr[R] y_i$ for all $i$ then
  \[
  \parinterpretation{\nop}{}(x_1, \ldots, x_n) \wbsimgr[T(R)] \parinterpretation{\nop}{}(y_1, \ldots, y_n)
  \]
  by guarded recursion.
  Suppose first that all $x_1, \ldots, x_n$ are of the form $\now^\kappa (x'_i)$ for $1 \leq i \leq n$, then
  \[
  \parinterpretation{\nop}{}(x_1, \ldots, x_n) = \now^\kappa \nop(x'_1, \ldots, x'_n),
  \]
  and since $x_i \wbsimgr[R] y_i$, there exist $m_i$ such that:
  \[
  y_i = (\nextstep)^{m_i} (\now^\kappa y'_i),
  \]
  with $R(x_i', y_i')$. Let $N = \max_i m_i$. Then:
  \begin{align*}
      \parinterpretation{\nop}{}(y_1, \ldots, y_n)
     & = \;  \parinterpretation{\nop}{}((\nextstep)^{m_1} (\now^\kappa y'_1), \ldots, (\nextstep)^{m_n} (\now^\kappa y'_n)) \\
     & =\; (\nextstep)^{N} (\now^\kappa (\nop (y'_1, \ldots, y'_n))).
  \end{align*}
  and so $\parinterpretation{\nop}{}(x_1, \ldots, x_n) \wbsimgr[R] \parinterpretation{\nop}{}(y_1, \ldots, y_n)$,
  by $N$ applications of Lemma~\ref{lem:wbsimgr:step}.
  The case of all $y_i$ of the form $y_i = \now^\kappa(y_i')$ is symmetric.

  Now suppose some of the $x_i$ are of the form $\step^\kappa(x_i')$ and similarly some $y_i$ are of the form $\step^\kappa(y_i')$.
  Then
\begin{align*}
 \parinterpretation{\nop}{}(x_1, \ldots, x_n)  & =  \step^\kappa(\tabs\tickA\kappa \parinterpretation{\nop}{}(x_1'', \ldots, x_n'')) \\
 \parinterpretation{\nop}{}(y_1, \ldots, y_n)  & =  \step^\kappa(\tabs\tickA\kappa \parinterpretation{\nop}{}(y_1'', \ldots, y_n'')),
\end{align*}
where $x_i'' = x_i$ if $x_i = \step^\kappa(x_i')$ and $x_i'' = \tapp{x_i'}$ if $x_i = \step^\kappa(x_i')$ and likewise
for $y_i''$. In either of the four possible cases we can prove that $\latbind\tickA\kappa (x_i'' \wbsimgr[R] y_i'')$,
using Lemma~\ref{lem:wbsimgr:step} in the two non-symmetric cases.
By guarded recursion we therefore get that:
\[
 \latbind\tickA\kappa (\parinterpretation{\nop}{}(x_1'', \ldots, x_n'') \wbsimgr[T(R)] \parinterpretation{\nop}{}(y_1'', \ldots, y_n'')),
\]
and so $\parinterpretation{\nop}{}(x_1, \ldots, x_n) \wbsimgr[T(R)] \parinterpretation{\nop}{}(y_1, \ldots, y_n)$.
\end{proof}

\begin{proof}[Proof of Theorem~\ref{thm:dist:law:setoid}]
 We must prove that the four laws of Definition~\ref{def:distlaw} hold. Three of these can be easily
 seen to hold already on the level of representatives. It remains to show that the natural transformation
 commutes with the multiplication $\mu$ of $\ciD$. This boils down to showing that for any term $t$
 with $n$ free variables and any vector $x_1,\ldots, x_n$ of elements in $\ciD\ciD TX$
 the following holds:
 \[
   \parinterpretation{t}{\ciD X}(\mu(x_1), \ldots, \mu(x_n)) \wbsim[T(R)] \mu(\parinterpretation{t}{\ciD \ciD X}(x_1, \ldots, x_n)),
 \]
 We prove this in the case of $t=\nop$ for any n-ary operation $\nop$. 
 The general case then follows by induction using transitivity of $\wbsim[T(R)]$ as follows 
\begin{align*}
 \parinterpretation{\nop}{\ciD X}(\parinterpretation{t_1}{\ciD X}(\mu(\vec x)), \ldots, \parinterpretation{t_n}{\ciD X}(\mu(\vec x))) 
 & \wbsim[T(R)] \parinterpretation{\nop}{\ciD X}(\mu(\parinterpretation{t_1}{\ciD \ciD X}(\vec x)), \ldots, \mu(\parinterpretation{t_n}{\ciD \ciD X}(\vec x))) \\
 & \wbsim[T(R)] \mu(\parinterpretation{\nop}{\ciD \ciD X}(\parinterpretation{t_1}{\ciD \ciD X}(\vec x), \ldots, \parinterpretation{t_n}{\ciD \ciD X}(\vec x))
\end{align*}
 where $\mu(x_1,\dots, x_m) = (\mu(x_1), \dots, \mu(x_m))$.

 Because we need to use guarded recursion we prove the statement for $t=\nop$ in the 
 case of the guarded delay monad. More precisely, we prove the following more general statement: Suppose 
 $x_i \wbsimgr y_i$ for all $i$ (where $\wbsimgr$ is the relation on $\grD\grD X$),  then
  \[
  \parinterpretation{\nop}{\grD X}(\mu(x_1), \ldots, \mu(x_n)) \wbsimgr[T(R)] \mu(\parinterpretation{\nop}{\grD \grD X}(y_1, \ldots, y_n)),
  \]
 for any operation $\nop$ in the signature.

 We prove this by guarded recursion. Suppose first that all $x_i$ are of the form $\now^\kappa(x_i')$. Then:
\begin{align*}
 \parinterpretation{\nop}{\grD X}(\mu(x_1), \ldots, \mu(x_n))
 & = \parinterpretation{\nop}{\grD X}(x_1', \ldots, x_n') \\
 & = \mu(\parinterpretation{\nop}{\grD\grD X}(x_1, \ldots, x_n)).
\end{align*}
The proof now follows from operations and $\mu$ preserving weak bisimilarity. If all $y_i$ are
of the form $\now^\kappa(y_i')$ then:
\begin{align*}
 \mu(\parinterpretation{\nop}{\grD\grD X}(y_1, \ldots, y_n)) & = \parinterpretation{\nop}{\grD X}(y_1', \ldots, y_n') \\
 & = \parinterpretation{\nop}{\grD X}(\mu(y_1), \ldots, \mu(y_n)).
\end{align*}
and so again the proof follows from $\mu$ and operations preserving weak bisimilarity.

In the remaining cases
\begin{align*}
 \parinterpretation{\nop}{\grD X} (\mu(x_1) , \ldots, \mu(x_n))
   & = \step^\kappa(\tabs\tickA\kappa (\parinterpretation{\nop}{\grD X}(\mu(x_1'), \ldots, \mu(x_n')))) \\
 \mu(\parinterpretation{\nop}{\grD\grD X}  (y_1, \ldots, y_n))
  & = \step^\kappa(\tabs\tickA\kappa (\mu(\parinterpretation{\nop}{\grD\grD X}(y_1', \ldots, y_n'))),
\end{align*}
where
\[
x_i' =
\begin{cases}
 x_i & \text{if } x_i = \now^\kappa(\now^\kappa(x_i'')) \\
 \tapp{x_i''} & \text{if } x_i = \step^\kappa(x_i'') \\
 \now^\kappa(\tapp{x_i''}) & \text{if } x_i = \now^\kappa(\step^\kappa(x_i'')),
\end{cases}
\]
and
\[
y_i' =
\begin{cases}
 y_i & \text{if } y_i = \now^\kappa(y_i'') \\
 \tapp{y_i''} & \text{if } y_i = \step^\kappa(y_i'').
\end{cases}
\]
It remains to show that in all these cases $\latbind\tickA\kappa(x_i' \wbsimgr y_i')$. This is done
using guarded recursion as follows.
\begin{itemize}
  \item If $x_i = \now^\kappa(\now^\kappa(x_i''))$ and $y_i = \now^\kappa(y_i'')$, then $x'_i \wbsimgr y'_i$  by assumption.
  \item If $x'_i = \now^\kappa(\now^\kappa(x_i''))$ and $y_i = \step^\kappa(y_i'')$ then
  the result follows from Lemma~\ref{lem:wbsimgr:step}.
  \item The case where $x_i = \step^\kappa(x_i'')$ and $y_i = \now^\kappa(y_i'')$ also follows from Lemma~\ref{lem:wbsimgr:step}.
  \item In the case of $x_i = \step^\kappa(x_i'')$ and $y_i = \step^\kappa(y_i'')$ it follows immediately from the definition of $x_i \wbsimgr y_i$ that $\latbind\tickA\kappa(x_i' \wbsimgr y_i')$.
  \item If $x_i = \now^\kappa(\step^\kappa(x_i''))$ and $y_i = \now^\kappa(y_i'')$, then by definition of $x_i \wbsimgr y_i$ we get $\step^\kappa(x_i'') \wbsimgr y_i''$. But then also
  $\latbind\tickA\kappa(\tapp{x_i''} \wbsimgr y_i'')$ (Lemma~\ref{lem:wbsimgr:step}) and therefore $\latbind\tickA\kappa(\now^\kappa(\tapp{x_i''}) \wbsimgr \now^\kappa(y_i''))$.
  \item If $x_i = \now^\kappa(\step^\kappa(x_i''))$ and $y_i = \step^\kappa(y_i'')$, then it follows from $x_i \wbsimgr y_i$ that there is an $n : \bb{N}$ and a $y_i'''$ such that $y_i = (\nextstep)^n \now^\kappa (y''')$, and $\step^\kappa(x_i'') \wbsimgr y_i'''$, which implies that $\latbind\tickA\kappa(\tapp{x_i''} \wbsimgr y_i''')$.
      Since $y_i = \step^\kappa(y_i'')$, we also know that $\tapp{y_i''} = (\nextstep)^{n-1} \now^\kappa (y''')$. Hence by definition
      \[
      \latbind\tickA\kappa(\now^\kappa(\tapp{x_i''}) \wbsimgr \tapp{y_i''})
      \]
      as required. \qedhere
\end{itemize}
\end{proof}

\section{Related Work}

M{\o}gelberg and Vezzosi~\cite{mogelberg2021two} study two combinations of the guarded delay monad $\grD$ with
the finite powerset monad $\Pfin$ expressed as a HIT in CCTT. They use these to show that applicative simulation is
a congruence for the untyped lambda calculus with finite non-determinism using denotational techniques. One
combination is the sum $\monadsum \Pfin\grD$, which is used for the case of may-convergence, and the other is
the composite $\grD\Pfin$ equipped with the parallel lifting, which is used for must-convergence. They observe that only the former is a monad. In this paper, we not only provide a more general study of such combinations, but also
suggest a way to remedy the situation in the latter case by considering weak bisimilarity.

Weak bisimilarity for the coinductive delay monad was first defined by Capretta~\cite{capretta2005general}.
M{\o}gelberg and Paviotti~\cite{Paviotti:FPC:journal} show that their embedding of FPC in guarded dependent type theory
respects weak bisimilarity and use that to prove an adequacy theorem up to weak bisimilarity.

Chapman et al.~\cite{quotientingDelay} observe that quotienting the coinductive delay monad by weak bisimilarity appears to not
yield a monad unless countable choice is assumed. Altenkirch et al.~\cite{altenkirch2017partiality} propose a solution to this problem by
constructing the quotient and the weak bisimilarity relation simultaneously, as a higher inductive-inductive type.
Chapman et al. themselves suggest a different solution, constructing the quotient as the
free $\omega$-cpo using an ordinary HIT. These
quotients have not (to the best of our knowledge) been studied in combination with other effects.

Interaction trees~\cite{Interaction:trees} are essentially monads of the form
$\forall\kappa . (\monadsum T{\grD}) (-)$ for $T$ an algebraic monad
generated by operations with no equations.
Much work has gone into building libraries for working with these up to
weak bisimilarity in Coq, and these allow for interaction trees to be used
for program verification. To our knowledge,
versions of interaction trees with equations between terms have not been considered.

As mentioned in the introduction, the guarded recursive delay monad has two benefits over the coinductive one:
Firstly, it has a fixed point operator of the type (rather than an iteration operator), which means that it allows for embedding languages
with recursion directly in type theory. In the coinductive case, one must either use some encoding of recursion using
the iteration operator, or prove that all constructions used are continuous. We believe this is a considerable burden
for higher order functions. The second advantage is that guarded recursion allows for also advanced notions of state
to be encoded, as shown recently by Sterling et al.~\cite{sterling-gratzer-birkedal:2022}. Neither the interaction trees nor
the quotiented delay monads appear to have these benefits.

\subsection*{Related Work on Monad Compositions}
The field of monad compositions in general has attracted quite a bit of attention lately. After Plotkin proved that there is no distributive law combining probability and non-determinism~\cite{VaraccaWinskel2006}, Klin and Salamanca~\cite{KlinSalamanca2018} studied impossible distributions of the powerset monad over itself, while Zwart and Marsden provided a general study on what makes distributive laws fail~\cite{Zwart2019}. Meanwhile, the initial study of monad compositions by Manes and Mulry~\cite{ManesMulry2007, ManesMulry2008} was continued by Parlant et al \cite{Dahlqvist2017, parlant-thesis, parlant2020}. In both the positive and the negative theorems on distributive laws in these papers, certain classes of equations were identified as causes for making or breaking the monad composition. Idempotence, duplication, and dropping variables came out as especially noteworthy types of equations, which the findings in this paper confirm.

Our study of the delay monad provides an interesting extension on the previous works, because of its non-standard algebraic structure given by delay algebras, and the fact that the delay monad is neither affine nor relevant, which are the main properties studied by Parlant et al.

\section{Conclusion and Future Work}

We have studied how both the guarded recursive and the coinductive version of the delay monad combine with other monads. After studying some specific examples and free combinations, we looked more generally at possible distributive laws of $T\grD \rightarrow \grD T$. We found two natural candidates for such distributive laws, induced by \emph{parallel} and \emph{sequential} lifting of operations on $T$. We showed that:
\begin{itemize}
  \item Sequential lifting provides a distributive law for monads presented by theories with balanced equations.
  \item There is no distributive law possible for monads with a binary operation that is commutative and idempotent over $\grD$, but this does not rule out a distributive law of such monads over $\ciD$.
  \item Parallel lifting does not define a distributive law, but it does define one \emph{up to weak bisimilarity}, for monads presented by theories with non-drop equations.
\end{itemize}

It is unfortunate that weak bisimilarity requires working with setoids, due to the quotient of $\ciD$ up to weak bisimilarity not being a monad~\cite{quotientingDelay}.
It is not clear how to adapt the solutions to this problem mentioned above~\cite{quotientingDelay,altenkirch2017partiality} to the guarded recursive setting.

This paper only considers the case of finite arity operations (except for state, which can be of any arity). Distributive laws for countable arity operations
such as countable non-deterministic choice are more difficult. In those cases sequential lifting seems an unnatural choice, not only because it does not interact
well with idempotency, but also because it introduces divergence even in the cases where there is an upper limit to the number of steps taken by the arguments.
Extending our parallel lifting operation to the countable case requires deciding whether all the countably many input operations are values, which is not possible
in type theory.

The results presented in this paper are formulated and proven in CCTT. It is natural to ask whether the results proven for the coinductive
delay monad $\ciD$ also hold for $\ciD$ considered as a monad on the category $\Set$ of sets. For some of the results proven in this paper
(Proposition~\ref{prop:nodistpowerset} and Example~\ref{ex:dist:law:idempotent})
both the statements and proofs can be read in $\Set$. These results can therefore easily be seen to hold in this setting. In many other cases,
our constructions use guarded recursion (e.g. the definitions of parallel and sequential lifting of operators). To lift these results to $\Set$,
one would need to redo the constructions and argue for their productivity. However, we believe that using guarded recursion is the natural way
to work with coinductive types and proofs. Another approach could therefore be to use guarded recursion as a language to reason about $\Set$.
This should be possible because the universe used to model clock irrelevant types in the extensional model of Clocked Type Theory~\cite{clottmodel}
classifies a category equivalent to $\Set$. We leave this as a direction for future research.

\bibliographystyle{alphaurl}
\bibliography{paper}

\newcommand{\etalchar}[1]{$^{#1}$}
\begin{thebibliography}{FGGW18}

\bibitem[ADK17]{altenkirch2017partiality}
Thorsten Altenkirch, Nils~Anders Danielsson, and Nicolai Kraus.
\newblock Partiality, revisited.
\newblock In {\em International Conference on Foundations of Software Science and Computation Structures}, pages 534--549. Springer, 2017.

\bibitem[AM13]{atkey13icfp}
Robert Atkey and Conor McBride.
\newblock Productive coprogramming with guarded recursion.
\newblock {\em ACM SIGPLAN Notices}, 48(9):197--208, 2013.

\bibitem[Bec69]{Beck1969}
Jon Beck.
\newblock Distributive laws.
\newblock In B.~Eckmann, editor, {\em Seminar on Triples and Categorical Homology Theory}, pages 119--140. Springer Berlin Heidelberg, 1969.

\bibitem[BGM17]{bahr2017clocks}
Patrick Bahr, Hans~Bugge Grathwohl, and Rasmus~Ejlers M{\o}gelberg.
\newblock The clocks are ticking: No more delays!
\newblock In {\em Proceedings of the 32nd Annual ACM/IEEE Symposium on Logic in Computer Science (LiCS)}, pages 1--12, 2017.

\bibitem[BH22]{mococa}
Patrick Bahr and Graham Hutton.
\newblock Monadic compiler calculation (functional pearl).
\newblock {\em Proceedings of the ACM on Programming Languages}, 6(ICFP), 2022.
\newblock \href {https://doi.org/10.1145/3547624} {\path{doi:10.1145/3547624}}.

\bibitem[BHM00]{benton2000}
Nick Benton, John Hughes, and Eugenio Moggi.
\newblock Monads and effects.
\newblock In {\em International Summer School on Applied Semantics}, pages 42--122. Springer, 2000.

\bibitem[BKMV22]{CubicalCloTT}
Magnus Baunsgaard~Kristensen, Rasmus~Ejlers Mogelberg, and Andrea Vezzosi.
\newblock Greatest hits: Higher inductive types in coinductive definitions via induction under clocks.
\newblock In {\em Proceedings of the 37th Annual ACM/IEEE Symposium on Logic in Computer Science (LiCS)}, pages 1--13, 2022.

\bibitem[BMSS12]{ToT}
Lars Birkedal, Rasmus~Ejlers M{\o}gelberg, Jan Schwinghammer, and Kristian St{\o}vring.
\newblock First steps in synthetic guarded domain theory: step-indexing in the topos of trees.
\newblock {\em Logical Methods in Computer Science}, 8(4), 2012.

\bibitem[Cap05]{capretta2005general}
Venanzio Capretta.
\newblock General recursion via coinductive types.
\newblock {\em Logical Methods in Computer Science}, 1, 2005.

\bibitem[CCHM18]{CTT}
Cyril Cohen, Thierry Coquand, Simon Huber, and Anders M{\"o}rtberg.
\newblock Cubical type theory: A constructive interpretation of the univalence axiom.
\newblock In {\em 21st International Conference on Types for Proofs and Programs (TYPES 2015)}. Schloss Dagstuhl-Leibniz-Zentrum fuer Informatik, 2018.

\bibitem[CH19]{CavalloHarper}
Evan Cavallo and Robert Harper.
\newblock Higher inductive types in cubical computational type theory.
\newblock {\em Proceedings of the ACM on Programming Languages}, 3(POPL):1--27, 2019.

\bibitem[Che11]{Cheng2011}
Eugenia Cheng.
\newblock Iterated distributive laws.
\newblock {\em Mathematical Proceedings of the Cambridge Philosophical Society}, 150(3):459--487, 2011.
\newblock \href {https://doi.org/10.1017/S0305004110000599} {\path{doi:10.1017/S0305004110000599}}.

\bibitem[Clo18]{clouston2018fitch}
Ranald Clouston.
\newblock Fitch-style modal lambda calculi.
\newblock In {\em International Conference on Foundations of Software Science and Computation Structures}, pages 258--275. Springer, 2018.

\bibitem[CUV19]{quotientingDelay}
James Chapman, Tarmo Uustalu, and Niccol{\`o} Veltri.
\newblock Quotienting the delay monad by weak bisimilarity.
\newblock {\em Mathematical Structures in Computer Science}, 29(1):67--92, 2019.

\bibitem[Dan12]{NAD:SIGPLAN:Not:2012}
Nils~Anders Danielsson.
\newblock Operational semantics using the partiality monad.
\newblock {\em SIGPLAN Not.}, 47(9):127--138, sep 2012.
\newblock \href {https://doi.org/10.1145/2398856.2364546} {\path{doi:10.1145/2398856.2364546}}.

\bibitem[Dan18]{Danielsson18}
Nils~Anders Danielsson.
\newblock Up-to techniques using sized types.
\newblock {\em Proceedings of the ACM on Programming Languages}, 2({POPL}):43:1--43:28, 2018.
\newblock \href {https://doi.org/10.1145/3158131} {\path{doi:10.1145/3158131}}.

\bibitem[DPS18]{Dahlqvist2017}
Fredrik Dahlqvist, Louis Parlant, and Alexandra Silva.
\newblock Layer by layer - combining monads.
\newblock In {\em Theoretical Aspects of Computing - {ICTAC} - 15th International Colloquium, Stellenbosch, South Africa, October 16-19, 2018, Proceedings}, pages 153--172, 2018.
\newblock \href {https://doi.org/10.1007/978-3-030-02508-3\_9} {\path{doi:10.1007/978-3-030-02508-3\_9}}.

\bibitem[EO10a]{escardo2010b}
Mart{\'\i}n Escard{\'o} and Paulo Oliva.
\newblock What sequential games, the tychonoff theorem and the double-negation shift have in common.
\newblock {\em Proceedings of the ACM SIGPLAN International Conference on Functional Programming (ICFP)}, 09 2010.
\newblock \href {https://doi.org/10.1145/1863597.1863605} {\path{doi:10.1145/1863597.1863605}}.

\bibitem[EO10b]{escardo_oliva_2010}
Martín Escardó and Paulo Oliva.
\newblock Selection functions, bar recursion and backward induction.
\newblock {\em Mathematical Structures in Computer Science}, 20(2):127–168, 2010.
\newblock \href {https://doi.org/10.1017/S0960129509990351} {\path{doi:10.1017/S0960129509990351}}.

\bibitem[EO17]{escardo2017}
Mart{\'\i}n Escard{\'o} and Paulo Oliva.
\newblock The herbrand functional interpretation of the double negation shift.
\newblock {\em The Journal of Symbolic Logic}, 82(2):590--607, 2017.

\bibitem[FGGW18]{Geuvers2018}
Dan Frumin, Herman Geuvers, L\'{e}on Gondelman, and Niels van~der Weide.
\newblock Finite sets in homotopy type theory.
\newblock In {\em Proceedings of the 7th ACM SIGPLAN International Conference on Certified Programs and Proofs (CPP)}, page 201–214, New York, NY, USA, 2018. Association for Computing Machinery.
\newblock \href {https://doi.org/10.1145/3167085} {\path{doi:10.1145/3167085}}.

\bibitem[Fio11]{FioreNotes}
M.~Fiore.
\newblock personal research notes, including a proof in {C}oq, February 2011.

\bibitem[Gau57]{gautam1957}
ND~Gautam.
\newblock The validity of equations of complex algebras.
\newblock {\em Archiv f{\"u}r mathematische Logik und Grundlagenforschung}, 3(3):117--124, 1957.
\newblock \href {https://doi.org/10.1007/BF01988052} {\path{doi:10.1007/BF01988052}}.

\bibitem[GKNB20]{gratzer2020multimodal}
Daniel Gratzer, GA~Kavvos, Andreas Nuyts, and Lars Birkedal.
\newblock Multimodal dependent type theory.
\newblock In {\em Proceedings of the 35th Annual ACM/IEEE Symposium on Logic in Computer Science (LiCS)}, pages 492--506, 2020.
\newblock \href {https://doi.org/10.1145/3373718.339473} {\path{doi:10.1145/3373718.339473}}.

\bibitem[Hed14]{Hedges2014}
Jules Hedges.
\newblock Monad transformers for backtracking search.
\newblock In {\em {\rm Proceedings 5th Workshop on} Mathematically Structured Functional Programming (MSFP)}, volume 153 of {\em Electronic Proceedings in Theoretical Computer Science}, pages 31--50. Open Publishing Association, 2014.
\newblock \href {https://doi.org/10.4204/EPTCS.153.3} {\path{doi:10.4204/EPTCS.153.3}}.

\bibitem[HLPP07]{hyland2007}
Martin Hyland, Paul~Blain Levy, Gordon Plotkin, and John Power.
\newblock Combining algebraic effects with continuations.
\newblock {\em Theoretical Computer Science}, 375(1-3):20--40, 2007.

\bibitem[HP06]{hyland2006discrete}
Martin Hyland and John Power.
\newblock Discrete lawvere theories and computational effects.
\newblock {\em Theoretical Computer Science}, 366(1-2):144--162, 2006.

\bibitem[HPP06]{hyland2006}
Martin Hyland, Gordon Plotkin, and John Power.
\newblock Combining effects: Sum and tensor.
\newblock {\em Theoretical Computer Science}, 357(1):70--99, 2006.
\newblock \href {https://doi.org/10.1016/j.tcs.2006.03.013} {\path{doi:10.1016/j.tcs.2006.03.013}}.

\bibitem[Jac10]{Jacobs2010}
Bart Jacobs.
\newblock Convexity, duality and effects.
\newblock In {\em Theoretical Computer Science}, pages 1--19. Springer Berlin Heidelberg, 2010.

\bibitem[Koc70]{kock1970}
Anders Kock.
\newblock Monads on symmetric monoidal closed categories.
\newblock {\em Archiv der Mathematik}, 21(1):1--10, 1970.

\bibitem[KS18]{KlinSalamanca2018}
Bartek Klin and Julian Salamanca.
\newblock Iterated covariant powerset is not a monad.
\newblock In {\em Proceedings 34th Conference on the Mathematical Foundations of Programming Semantics, {MFPS} 2018}, 2018.

\bibitem[Law63]{Lawvere1963}
Francis~William Lawvere.
\newblock Functorial semantics of algebraic theories.
\newblock {\em Proceedings of the National Academy of Sciences of the United States of America}, 50(5):869--872, 1963.

\bibitem[Lin66]{Linton1966}
Fred Linton.
\newblock Some aspects of equational categories.
\newblock In {\em Proceedings of the Conference on Categorical Algebra}, pages 84--94. Springer, 1966.

\bibitem[Man76]{Manes1976}
Ernie Manes.
\newblock {\em Algebraic theories}, volume~26.
\newblock Springer, 1976.

\bibitem[Mee86]{Meertens1986}
Lambert Meertens.
\newblock Algorithmics, towards programming as a mathematical activity.
\newblock In {\em Mathematics and Computer Science: Proceedings of the CWI Symposium, November 1983}, CWI monographs, pages 289--334. North-Holland, 1986.

\bibitem[ML84]{MartinLof:84}
Per Martin-L\"of.
\newblock {\em Intuitionistic Type Theory}.
\newblock Bibliopolis, Napoli, 1984.

\bibitem[MM07]{ManesMulry2007}
Ernie Manes and Philip Mulry.
\newblock Monad compositions {I}: general constructions and recursive distributive laws.
\newblock {\em Theory and Applications of Categories}, 18:172--208, 04 2007.

\bibitem[MM08]{ManesMulry2008}
Ernie Manes and Philip Mulry.
\newblock Monad compositions {II}: Kleisli strength.
\newblock {\em Mathematical Structures in Computer Science}, 18(3):613–643, 2008.
\newblock \href {https://doi.org/10.1017/S0960129508006695} {\path{doi:10.1017/S0960129508006695}}.

\bibitem[MMV20]{clottmodel}
Bassel Mannaa, Rasmus~Ejlers M{\o}gelberg, and Niccol{\`o} Veltri.
\newblock {Ticking clocks as dependent right adjoints: Denotational semantics for clocked type theory}.
\newblock {\em Logical Methods in Computer Science}, 16, 2020.

\bibitem[MP19]{Paviotti:FPC:journal}
Rasmus~E M{\o}gelberg and Marco Paviotti.
\newblock Denotational semantics of recursive types in synthetic guarded domain theory.
\newblock {\em Mathematical Structures in Computer Science}, 29(3):465--510, 2019.

\bibitem[MV21]{mogelberg2021two}
Rasmus~Ejlers M{\o}gelberg and Andrea Vezzosi.
\newblock Two guarded recursive powerdomains for applicative simulation.
\newblock In {\em Proceedings of the 37th Conference on Mathematical Foundations of Programming Semantics (MFPS)}, volume 351 of {\em {EPTCS}}, pages 200--217, 2021.
\newblock \href {https://doi.org/10.4204/EPTCS.351.13} {\path{doi:10.4204/EPTCS.351.13}}.

\bibitem[MZ24]{MogelbergZwart2024}
Rasmus~Ejlers M{\o}gelberg and Maaike Zwart.
\newblock {What Monads Can and Cannot Do with a Bit of Extra Time}.
\newblock In {\em 32nd EACSL Annual Conference on Computer Science Logic (CSL 2024)}, volume 288 of {\em Leibniz International Proceedings in Informatics (LIPIcs)}, pages 39:1--39:18. Schloss Dagstuhl -- Leibniz-Zentrum f{\"u}r Informatik, 2024.
\newblock \href {https://doi.org/10.4230/LIPIcs.CSL.2024.39} {\path{doi:10.4230/LIPIcs.CSL.2024.39}}.

\bibitem[Nak00]{Nakano:Modality}
Hiroshi Nakano.
\newblock A modality for recursion.
\newblock In {\em Proceedings of the 15th Annual ACM/IEEE Symposium on Logic in Computer Science (LiCS)}, pages 255--266. IEEE, 2000.

\bibitem[Par20]{parlant-thesis}
Louis Parlant.
\newblock {\em Monad Composition via Preservation of Algebras}.
\newblock PhD thesis, UCL (University College London), 2020.

\bibitem[PMB15]{paviottiPCF}
Marco Paviotti, Rasmus~Ejlers M{\o}gelberg, and Lars Birkedal.
\newblock A model of {PCF} in guarded type theory.
\newblock {\em Electronic Notes in Theoretical Computer Science}, 319:333--349, 2015.

\bibitem[PP02]{plotkin2002}
Gordon Plotkin and John Power.
\newblock Notions of computation determine monads.
\newblock In {\em International Conference on Foundations of Software Science and Computation Structures}, pages 342--356. Springer, 2002.

\bibitem[PRSW20]{parlant2020}
Louis Parlant, Jurriaan Rot, Alexandra Silva, and Bas Westerbaan.
\newblock Preservation of equations by monoidal monads.
\newblock In {\em 45th International Symposium on Mathematical Foundations of Computer Science (MFCS 2020)}. Schloss Dagstuhl-Leibniz-Zentrum f{\"u}r Informatik, 2020.

\bibitem[SGB22]{sterling-gratzer-birkedal:2022}
Jonathan Sterling, Daniel Gratzer, and Lars Birkedal.
\newblock Denotational semantics of general store and polymorphism.
\newblock Unpublished manuscript, July 2022.

\bibitem[SMZ{\etalchar{+}}25]{Stassen2024}
Philipp Jan~Andries Stassen, Rasmus~Ejlers Møgelberg, Maaike Zwart, Alejandro Aguirre, and Lars Birkedal.
\newblock Modelling recursion and probabilistic choice in guarded type theory.
\newblock {\em Proceedings of the ACM on Programming Languages}, 9(POPL):1417 -- 1445, 2025.
\newblock \href {https://doi.org/10.1145/3704884} {\path{doi:10.1145/3704884}}.

\bibitem[{Uni}13]{hottbook}
The {Univalent Foundations Program}.
\newblock {\em Homotopy Type Theory: Univalent Foundations of Mathematics}.
\newblock \url{https://homotopytypetheory.org/book}, Institute for Advanced Study, 2013.

\bibitem[Vez]{AgdaCubical}
Andrea Vezzosi.
\newblock https://github.com/agda/guarded, last accessed 2 june 2025.

\bibitem[VV21]{VVMFPS2021}
Niccol{\`{o}} Veltri and Niels Voorneveld.
\newblock Inductive and coinductive predicate liftings for effectful programs.
\newblock In {\em Proceedings 37th Conference on Mathematical Foundations of Programming Semantics (MFPS)}, volume 351 of {\em {EPTCS}}, pages 260--277, 2021.
\newblock \href {https://doi.org/10.4204/EPTCS.351.16} {\path{doi:10.4204/EPTCS.351.16}}.

\bibitem[VW06]{VaraccaWinskel2006}
Daniele Varacca and Glynn Winskel.
\newblock Distributing probability over non-determinism.
\newblock {\em Mathematical Structures in Computer Science}, 16(1):87--113, 2006.

\bibitem[XZH{\etalchar{+}}19]{Interaction:trees}
Li-yao Xia, Yannick Zakowski, Paul He, Chung-Kil Hur, Gregory Malecha, Benjamin~C. Pierce, and Steve Zdancewic.
\newblock Interaction trees: Representing recursive and impure programs in {Coq}.
\newblock {\em Proceedings of the ACM on Programming Languages}, 4(POPL), 2019.
\newblock \href {https://doi.org/10.1145/3371119} {\path{doi:10.1145/3371119}}.

\bibitem[ZM19]{Zwart2019}
Maaike Zwart and Dan Marsden.
\newblock No-go theorems for distributive laws.
\newblock In {\em Proceedings of the 34th Annual ACM/IEEE Symposium on Logic in Computer Science (LiCS)}, pages 1--13, 2019.
\newblock \href {https://doi.org/10.1109/LICS.2019.8785707} {\path{doi:10.1109/LICS.2019.8785707}}.

\bibitem[Zwa]{AgdaFiles}
Maaike Zwart.
\newblock Agda files for lmcs paper \emph{What monads can and cannot do with a few extra pages}.
\newblock \href {https://doi.org/10.5281/zenodo.15593221} {\path{doi:10.5281/zenodo.15593221}}.

\end{thebibliography}

% !TEX root = paper.tex

\newcommand{\arity}{\sym{ar}}
\newcommand{\Sig}{\Sigma}

\newcommand{\Eq}{\sym{E}}
\newcommand{\Term}{\sym{Tm}}
\newcommand{\Fin}[1]{\overline{#1}}
\newcommand{\rhs}{\sym{rhs}}
\newcommand{\lhs}{\sym{lhs}}

\newcommand{\conVar}{\sym{var}}
\newcommand{\conTerm}{\sym{tm}}
\newcommand{\conFlat}{\sym{flt}}
\newcommand{\conEq}{\sym{eq}}
\newcommand{\trunc}{\sym{trunc}}
\newcommand{\isSet}{\sym{isSet}}
\newcommand{\summ}{\sym{sum}(m)}
\newcommand{\summk}{\sym{sum}(\capp m)}
\newcommand{\ink}{\sym{in}_{m,k}}
\newcommand{\inkkappa}{\sym{in}_{\capp m,k}}
\newcommand{\inkkk}{\sym{in}_{\capp m,\capp k}}

\newcommand{\uVar}{u_\sym{var}}
\newcommand{\uTerm}{u_\sym{tm}}
\newcommand{\uFlat}{u_\sym{flt}}
\newcommand{\uEq}{u_\sym{eq}} 

\newcommand{\GAP}{\hspace{1cm}}

\newcommand{\I}{\mathbb{I}}

\newpage

\appendix

\section{Encoding algebraic theories in CCTT}\label{subsec:encoding}
\label{app:encoding:alg:th}

This appendix shows how we encode the notion of algebraic theory
in CCTT, construct free functors as HITs, and proves Theorem~\ref{prop:free:monad:cirr}.

The type of algebraic theories can be defined as an element of a large universe written informally
as follows.

\begin{defi}
 An algebraic theory comprises
 \begin{itemize}
\item A signature type indexed by arities $\Sig : \N \to \USet$
\item A type of equations indexed by the number of free variables $\Eq : \N \to \USet$
\item Two maps giving the left hand side and right hand side of an equation respectively
\begin{align*}
 \lhs, \rhs & : (n : \N) \to \Eq(n) \to \Term(n)
\end{align*}
\end{itemize}
\end{defi}

Where $\Term(n)$ is the set of terms in context $\Fin{n}$, the type of numbers $0, \dots, n-1$. Both of these can
be defined in the standard type theoretic way. Recall that $\USet$ is the universe of small clock-irrelevant hsets.
Both $\Term(n)$ and $\Fin{n}$ are clock irrelevant because $\Sig(n)$ is~\cite[Theorem~5.3]{CubicalCloTT}.

The free monad on an algebraic signature $A = (\Sig, \arity, \Eq, \lhs,\rhs)$ is defined to map $X : \USet$
to a higher inductive type. The natural definition for this would be to have a constructor for each operation in $\Sig$
and a path for each element of $\Eq$, plus constructors and paths to set-truncate $T(X)$.
However, this would require recursion in the boundary condition of the equalities
to interpret the terms given by $\lhs$ and $\rhs$, and this is not allowed in the rule format for HITs in CCTT. Instead we
use a constructor for each term of the theory and add equations to enforce flattening of terms, as expressed in the following
definition which follows the format of HITs in CCTT~\cite{CubicalCloTT}.

\begin{defi}
 The free monad given by the signature $A$ applied to $X : \USet$ is the HIT $T(X)$ given by the following constructors
\begin{align*}
 \conVar & : X \to T(X) \\
 \conTerm & : (n : \N) \to (\tau : \Term(n)) \to (\Fin{n} \to T(X)) \to T(X) \\
 \conFlat & : (n, m : \N) \to (\sigma : \Sig(n)) 
\to (\tau : \Fin{n} \to  \Term(m)) 
 \to (m \to T(X)) \to \I \to T(X) \\
\conEq & : (n : \N) \to (e : \Eq(n)) \to (\Fin{n} \to T(X)) \to \I \to T(X)
\end{align*}
These satisfy the following boundary conditions (writing $\jeq$ for judgemental equality)
\begin{align*}
 \conFlat\, n \, m\, \sigma\,\tau\,\rho\, 0 & \jeq \conTerm\,m \, (\sigma[\tau]) \,\rho \\
 \conFlat\, n\, m \, \sigma\, \tau\,\rho\, 1 & \jeq \conTerm\, n \, \sigma\,(\lambda k : \Fin{n}. \conTerm\,m\, (\tau \,k) \, \rho) \\
 \conEq\, n \, e \, \rho\, 0 & \jeq \conTerm \, n \, (\lhs\,n \,e)\, \rho \\
 \conEq\, n \, e \, \rho\, 1 & \jeq \conTerm \, n \, (\rhs\,n \,e)\, \rho
\end{align*}
where $\sigma[\tau]$ is the result of substituting the elements of $\tau$ for the free variables of $\sigma$.
To these should be added standard constructors and equalities for set truncation~\cite{CubicalCloTT},
which we omit.
\end{defi}

\subsection{The free monad commutes with clock quantification}

\begin{figure*}[b]
 \label{fig:induction:under:clocks}
\hspace{-.5cm}
\begin{minipage}{\linewidth}
\begin{center}
\textbf{Hypotheses}
\begin{description}
\item[A type] {\small $\istype{\Gamma, x : \forall\kappa. T(\capp X)}{B(x)}$}
\item[Terms]
{\small
\begin{align*}
 \hastype{\Gamma, x : \forall\kappa. (\capp X)&}{\uVar}{B(\lambda\kappa. \conVar (\capp x))} \\
 \hastype{\Gamma, n : \forall\kappa. \N, \tau : \forall\kappa . \Term(\capp n), \\ \rho : \forall\kappa . (\Fin{\capp n}\to T(\capp X)),  \\
 \rho' : (k : \forall\kappa . \Fin{\capp n})\to B(\lambda\kappa . \capp\rho (\capp k)) &}{\uTerm}{B(\lambda\kappa. \conTerm (\capp n)(\capp \tau)(\capp\rho))} \\
 \hastype{\Gamma, n, m : \forall\kappa. \N, \sigma : \forall\kappa . \Sig(\capp n), \\
 \tau : \forall\kappa . (\Fin{\capp n} \to \Term(\capp m)), \\ \rho : \forall\kappa . (\capp m \to T(\capp X)), \\
 \rho' : (k : \forall\kappa . \capp m) \to B(\lambda\kappa . \capp\rho(\capp k)), i : \I
 &}{\uFlat}{B(\lambda\kappa. \conFlat (\capp n)(\capp m)(\capp \sigma)(\capp\tau)(\capp \rho)\, i)} \\
 \hastype{\Gamma, n : \forall\kappa. \N, e : \forall\kappa . \Eq(\capp n), \\ \rho : \forall\kappa . (\Fin{\capp n}\to T(\capp X)),  \\
 \rho' : (k : \forall\kappa . \Fin{\capp n})\to B(\lambda\kappa . \capp\rho (\capp k)), i : \I &}{\uEq}{B(\lambda\kappa. \conEq (\capp n)(\capp e)(\capp \rho) \,i)} 
\end{align*}
}

\item[Satisfying]
{\small
\begin{align*}
 \uFlat(n, \sigma, m, \tau, \rho, \rho', 0) & \jeq \uTerm(\lambda\kappa. (\capp m), \lambda\kappa . (\capp \sigma)[\capp \tau], \rho, \rho') \\
 \uFlat(n, \sigma, m, \tau, \rho, \rho', 1) & \jeq \uTerm(n, \sigma, \lambda\kappa . \lambda k . \conTerm\,(\capp m)\,(\capp \tau\, k)\,(\capp \rho) ,% \\
\lambda k. \uTerm((\capp m), \lambda\kappa . \capp \tau(\capp k), \rho, \rho') \\
 \uEq(n, e, \rho, \rho', 0) & \jeq \uTerm(n, (\lambda\kappa . \lhs(\capp e)), \rho, \rho') \\
 \uEq(n, e, \rho, \rho', 1) & \jeq \uTerm(n, (\lambda\kappa . \rhs(\capp e)), \rho, \rho')
\end{align*}
}
\end{description}
\textbf{Conclusion}
\begin{description}
\item[A section] {\small$\hastype{\Gamma, x : \forall\kappa. T(\capp X)}{t(x)}{B(x)}$}
\item[Satisfying]
{\small
\begin{align*}
 t(\lambda\kappa. \conVar (\capp x)) & \jeq \uVar(x) \\
 t(\lambda\kappa. \conTerm (\capp n)(\capp \tau)(\capp\rho)) & \jeq \uTerm(n, \tau, \rho, \lambda k. t(\lambda\kappa . \capp\rho\, (\capp k))) \\
 t(\lambda\kappa. \conFlat (\capp n)(\capp m)(\capp \sigma)(\capp\tau)(\capp \rho)\, i)) & \jeq \uFlat(n, m, \sigma, \tau, \rho, \lambda k . t(\lambda\kappa . \capp \rho(\capp k)), i) \\
 t(\lambda\kappa. \conEq (\capp n)(\capp e)(\capp \rho) \,i) & \jeq \uEq(n, e, \rho, \lambda k. t(\lambda\kappa . \capp \rho(\capp k)), i)
\end{align*}
}
\end{description}
\end{center}
\end{minipage}
 \caption{The principle of induction under clocks for $T(X)$. We omit the hypotheses concerning set truncation, which are standard~\cite{CubicalCloTT}. We assume given
 an algebraic signature and an $X : \forall\kappa . \USet$ both in context $\Gamma$. Note that the equalities are judgemental (denoted $\jeq$)}
\end{figure*}

We now prove Theorem~\ref{prop:free:monad:cirr}. For this we need the principle of induction under clocks~\cite{CubicalCloTT} which in
the case of $T(X)$ specialises to the principle described in Figure~2.

First recall that the canonical map $f : T(\forall\kappa . \capp X) \to \forall\kappa . T(\capp X)$ defined using functoriality of $T$ satisfies the clauses
\begin{align*}
 f(\conVar \, x) & = \lambda\kappa . \conVar(\capp x) \\
 f(\conTerm \, n\, \tau\, \rho ) & = \lambda\kappa . (\conTerm \, n \, \tau \, (\lambda k.\capp{f(\rho(k))})) \\
 f(\conFlat \, n\,m \,\sigma\, \tau\, \rho\, i) & = \lambda\kappa . (\conFlat \, n \, m\,\sigma \, \tau \, (\lambda k.\capp{f(\rho(k))}) \,i) \\
 f(\conEq \, n\, e \, \rho\, i) & = \lambda\kappa . (\conEq \, n \, e \, (\lambda k.\capp{f(\rho(k))}) \,i)
\end{align*}
(omitting the clauses for set truncation, which are standard~\cite{CubicalCloTT}). The inverse to $f$ can be constructed using the principle of induction under clocks.
We write the definition using clauses as follows
\begin{align*}
 g(\lambda\kappa. \conVar (\capp x)) & \jeq \conVar(x) \\
 g(\lambda\kappa. \conTerm (\capp n)(\capp \tau)(\capp\rho)) & \jeq \conTerm(\cappc n, \cappc \tau, \lambda k. g(\lambda\kappa . \capp\rho\, k)) \\
 g(\lambda\kappa. \conFlat (\capp n)(\capp m)(\capp \sigma)(\capp\tau)(\capp \rho)\, i)) 
  &\jeq \conFlat(n [\kappa_0], m [\kappa_0], \sigma [\kappa_0], \tau [\kappa_0], \lambda k . g(\lambda\kappa . \capp \rho\, k), i) \\
 g(\lambda\kappa. \conEq (\capp n)(\capp e)(\capp \rho) \,i)  & \jeq \conEq(\cappc n, \cappc e, \lambda k. g(\lambda\kappa . \capp \rho\,k), i)
\end{align*}
Here $\clockc$ is the clock constant, which we will assume. In the cases of the recursive calls, $k$ will have type $\Fin{\cappc n}$ and $\capp\rho$ expects
an input of type $\Fin{\capp n}$. These types are equivalent, since $\N$ is clock-irrelevant. We leave this equivalence implicit.

We must show that $g$ satisfies the boundary condition, which we just do in a single case. The rest are similar.
\begin{align*}
 g(\lambda\kappa. \conEq (\capp n)(\capp e)(\capp \rho) \,0) 
 & \jeq \conEq(\cappc n, \cappc e, \lambda k. g(\lambda\kappa . \capp \rho\,k), 0) \\
 & \jeq \conTerm(\cappc n, \lhs(\cappc e), \lambda k. g(\lambda\kappa . \capp\rho\, k)) \\
 & \jeq g(\lambda\kappa. \conTerm (\capp n)(\lhs(\capp e))(\capp\rho))
\end{align*}
Again, we omit the standard cases of constructors and boundary cases for set truncation~\cite{CubicalCloTT}.

It is easy to show that $g\circ f$ is the identity by induction. For example,
\begin{align*}
 g(f(\conTerm \, n\, \tau\, \rho)) & \jeq g(\lambda\kappa . (\conTerm \, n \, \tau \, (\lambda k.\capp{f(\rho(k))}))) \\
 & \jeq g(\lambda\kappa . (\conTerm \, n \, \tau \, (\capp{(\lambda\kappa'.\lambda k.\capp[\kappa']{f(\rho(k))})}))) \\
 & \jeq \conTerm(n, \tau, \lambda k. g(\lambda\kappa . \capp{f(\rho(k))})) \\
 & \jeq \conTerm(n, \tau, \lambda k. g(f(\rho(k)))) \\
 & \jeq \conTerm(n, \tau, \rho)
\end{align*}
using induction in the last equality.

Proving $f \circ g$ the identity requires another use of induction under clocks. We just show one case.
\begin{align*}
 f(g&(\lambda\kappa. \conTerm (\capp n)(\capp \tau)(\capp\rho))) \\
 & \jeq f(\conTerm(\cappc n, \cappc \tau, \lambda k. g(\lambda\kappa . \capp\rho\, k))) \\
 & \jeq \lambda\kappa . \conTerm(\cappc n, \cappc \tau, \lambda k. \capp{f(g(\lambda\kappa' . \capp[\kappa']\rho\, k))}) \\
 & = \lambda\kappa . \conTerm(\cappc n, \cappc \tau, \lambda k. \capp{(\lambda\kappa' . \capp[\kappa']\rho\, k)}) \\
 & \jeq \lambda\kappa . \conTerm(\capp n, \capp \tau, \capp\rho)
\end{align*}
using the induction hypothesis and, in the last step, clock irrelevance of $\N$ and $\Term{(\capp n)}$.

We must show that this definition satisfies the boundary conditions. Since the target type is a path type of $\forall\kappa . T(\capp X)$, which is an hset, these
conditions hold up to path equality. To force them to hold up to judgemental equality, we can change the term using a composition, similarly to
the proof of Theorem~5.3 of~\cite{CubicalCloTT}.

\end{document}